\documentclass[authoryear]{elsarticle-mod}

\newcommand{\C}{\mathcal{C}}

\newcommand{\I}{\mathcal{I}} \newcommand{\J}{\mathcal{J}}

 \newcommand{\R}{\mathcal{R}}
 
 \newcommand{\V}{\mathcal{V}}

\newcommand{\ol}[1]{\overline{#1}}                

\newcommand{\set}[1]{\{#1\}}                      

\newcommand{\tup}[1]{\ensuremath{\langle #1\rangle}}            

\newcommand{\replfunc}[3]{\ensuremath{{#1}\frac{#3}{\raisebox{0.5 ex}{$\scriptstyle{#2}$}}}}

\usepackage{stix}

\usepackage{url}
\usepackage{xspace}
\usepackage{hyperref}
\usepackage{array}

\usepackage{latexsym}

\usepackage{amsmath}
\usepackage{amssymb}
\usepackage{amsfonts}

\usepackage{stackengine}

\usepackage[shortlabels]{enumitem}

\usepackage{listings}
\usepackage[utf8]{inputenc}

\usepackage[obeyFinal]{todonotes}
\usepackage{totcount}

\definecolor{snow}{rgb}{0.93,0.91,0.91}
\newcommand{\graybox}[1]{%
  \begin{lrbox}{0}\begin{tabular}{@{}l@{}}#1\end{tabular}\end{lrbox}%
  \setlength{\fboxsep}{0ex}
  {\centering\fcolorbox{snow}{snow}{\usebox{0}}}}

\newcommand{\nullvalue}{\textbf{\textsc{\textsf{N}}}\xspace}

\renewcommand{\i}{\emph{(i)}\xspace}

\DeclareMathOperator*{\RAJoin}{\bowtie}

\newcommand{\subR}[1][\R]{\ensuremath{\widetilde{#1}}}

\newcommand{\PowerSet}{\ensuremath{\wp}}

\newcommand{\intConv}[3]{\ensuremath{{#2}^{#1}\ifthenelse{\equal{#3}{}}{}{_{(#3)}}}\xspace}
\newcommand{\IRA}[2][]{\intConv{\PowerSet}{#2}{#1}}
\newcommand{\InRA}[2][]{\intConv{\nullvalue}{#2}{#1}}
\newcommand{\IFOLe}[2][]{\intConv{\hole}{#2}{#1}}

\newcommand{\IFOLnin}[2][]{\intConv{{\scriptscriptstyle\text{ni}}}{#2}{#1}}

\newcommand{\FOLetoFOLd}{\ensuremath{\Omega_f}\xspace}
\newcommand{\FOLdtoFOLe}{\ensuremath{\Omega_{f}^{-1}}\xspace}

\newcommand{\hole}{\ensuremath{\varepsilon}\xspace}
\newcommand{\FOLfam}[2]{\ensuremath{\mathcal{F\!O\!L}\ifthenelse{\equal{#1}{}}{}{^{#1}}\ifthenelse{\equal{#2}{}}{}{_{#2}}}\xspace}
\newcommand{\FOL}[1][]{\FOLfam{}{#1}}
\newcommand{\FOLe}[1][]{\FOLfam{\hole}{#1}}

\newcommand{\true}{\ensuremath{\textbf{\textsf{true}}}\xspace}
\newcommand{\false}{\ensuremath{\textbf{\textsf{false}}}\xspace}
\newcommand{\foleq}[2]{\ensuremath{{\foleqp}({#1},{#2})}\xspace}
\newcommand{\foleqp}{\ensuremath{=}\xspace}

\newcommand{\nullify}[3]{\ensuremath{\Pi_{#1}{#3}({#2})}\xspace}

\newcommand{\RA}[1][]{\ensuremath{\mathcal{R\!\!A}_{{#1}}}\xspace}
\newcommand{\nRA}[1][]{\ensuremath{\mathcal{R\!\!A}^{\nullvalue}_{{#1}}}\xspace}
\newcommand{\select}[1]{{\ensuremath{\mathrel{\sigma_{#1}}}\xspace}}
\newcommand{\project}[2]{{\ensuremath{\mathrel{\pi_{#1}{#2}}}\xspace}}
\newcommand{\isnull}{\ensuremath{\mathsf{isNull}}\xspace}
\newcommand{\isnotnull}{\ensuremath{\mathsf{isNotNull}}\xspace}
\newcommand{\singleton}[1]{{\ensuremath{\tup{#1}}\xspace}}

\newcommand{\anstup}[2][\tau]{\ensuremath{{#1}[{#2}]}\xspace}

\newcommand{\eval}[2][\I]{\ensuremath{[#2]_{#1}}\xspace}

\newcommand{\sql}{\texttt{SQL}\xspace}
\newcommand{\fosql}{\texttt{SQL}\ensuremath{^{\mathcal F\!O}}\xspace}

\newtheorem{theorem}{Theorem} 
\newtheorem{lemma}{Lemma} 
\newtheorem{proposition}[theorem]{Proposition} 
 
\newdefinition{definition}{Definition} 
\newdefinition{example}{Example} 
\newproof{proof}{Proof}

\begin{document} 

\begin{frontmatter}

\title{Relational Algebra and Calculus \\ with SQL Null Values}
\author[1]{Enrico Franconi}
\ead{franconi@inf.unibz.it}
\author[1]{Sergio Tessaris}
\ead{tessaris@inf.unibz.it}
\affiliation[1]{organization={KRDB Research Centre for Knowledge and Data}, 
                addressline={Free University of Bozen-Bolzano}, 
                country={Italy}}

\begin{abstract}
	The logic of nulls in databases has been subject of investigation since their introduction in Codd's Relational Model, which is the foundation of the SQL standard. 
	We show a logical characterisation of a first-order fragment of SQL with null values, by first focussing on a simple extension with null values of standard relational algebra, which captures exactly the SQL fragment, and then proposing two different domain relational calculi, in which the null value is a term of the language but it does not appear as an element of the semantic interpretation domain of the logics. In one calculus, a relation can be seen as a set of partial tuples, while in the other (equivalent) calculus, a relation is horizontally decomposed as a set of relations each one holding regular total tuples.
	We extend Codd's theorem by proving the equivalence of the relational algebra with both domain relational calculi in presence of SQL null values.
\end{abstract}

\end{frontmatter}



\section{Relational Databases and SQL Null Values}
\label{sec:intro} 

This paper studies how missing information is represented and handled in the SQL relational database standard by means of null values. Since their inception, SQL null values have been at the centre of long discussions about their real meaning and their formal semantics -- good summaries of the discussions and critiques can be found in \citep{grant:2008,codd_missing_1986,rubinson_ontological_2014}. The SQL standard committee itself discussed whether to change the semantics of SQL null value~\citep{ANSI-SQL-1975,cannan_proposal_1987}. Our interest is not to propose alternative ways to represent null values and missing information in relational databases, but to exactly characterise SQL null values in the way they are defined in the SQL standard. Most of the theoretical research on null values in relational databases has been focused on different notions of semantics which diverge from the behaviour of null values in SQL. In spite of the fact that these works have their merits and provide a well founded characterisation of incomplete information in databases, they don't provide an explanation of null values according to the SQL standard.
As a matter of fact, to the best of our knowledge, while there has been some recent formalisation of a relational algebra and calculus dealing with SQL null values, there has been no attempt yet to formalise a logic extending the standard domain relational calculus with SQL null values, where the null value acts a syntactic marker for missing information in the language and therefore it does not appear as an element of the semantic interpretation domain.
In this paper, we want to understand the exact model-theoretic semantics of null values in SQL, and characterise it by an extended relational algebra and domain calculus, just like the null-free well-behaving core fragment of SQL has been characterised by standard relational algebra and domain calculus. 

To start with this endeavour, let's consider the following example SQL database as a ``black box'' and let's check its behaviour with respect to some SQL queries. The database contains a table \texttt{r} having a tuple with a null value (denoted by the symbol \nullvalue here):

\vspace{1ex}\noindent
\begin{minipage}{.075\textwidth}
~
\end{minipage}
\begin{minipage}{.6\textwidth}
{\footnotesize
\begin{verbatim}
CREATE TABLE r ( c1 TEXT NOT NULL, c2 TEXT NULL );
INSERT INTO r (c1, c2) VALUES ('a', 'a'), ('b', NULL);
\end{verbatim}
}
\end{minipage}
\begin{minipage}{.3\textwidth} 
{\footnotesize
\hfill
\stackanchor{\texttt{r}~~~~~~~~~~~~}{
	$ 
	\begin{array}{|@{\hspace{1ex}}wc{1em}@{\hspace{1ex}}|@{\hspace{1ex}}wc{1em}@{\hspace{1ex}}|}
	\hline
		\texttt{c1}          & \texttt{c2}          \\ 
		\hline
		\texttt{a}          & \texttt{a}          \\ 
		\texttt{b}          & \nullvalue          \\ 
		\hline
	\end{array}$}
}
\end{minipage}

\vspace{1ex}
Let's first consider the \emph{self-join} SQL query asking for the tuples in \texttt{r} with each column's value equal to itself.  Our intuition suggests that this query should behave as the identity query; however this is not the case in SQL:

\vspace{1ex}\noindent
\begin{minipage}{.075\textwidth}
\hspace{1em}(1)
\end{minipage}
\begin{minipage}{.6\textwidth}
{\footnotesize
\begin{verbatim}
SELECT c1, c2 FROM r WHERE c1 = c1 AND c2 = c2;
\end{verbatim}
}
\end{minipage}
\begin{minipage}{.3\textwidth}
\hfill 
{\footnotesize
	$\stackanchor{\text{\tiny{SQL}}}{\text{\tiny{semantics}}}\!\dashrightarrow\  
	\begin{array}{|@{\hspace{1ex}}wc{1em}@{\hspace{1ex}}|@{\hspace{1ex}}wc{1em}@{\hspace{1ex}}|}
	\hline
		\texttt{c1}          & \texttt{c2}          \\ 
		\hline
		\texttt{a}          & \texttt{a}          \\ 
		\hline
	\end{array}$
}
\end{minipage}

\vspace{1ex}
\noindent
The SQL query (1) above does not yield the table \texttt{r} itself: it returns just the tuples not containing a null value. If null values were treated as standard database constants, the direct translation in the standard relational algebra \RA of the above SQL query would be equivalent to the \emph{identity} expression for \texttt{r}: $\select{c1=c1}{\select{c2=c2}{r}}$, giving as an answer the table \texttt{r} itself, namely the tuples $\{\langle a,b\rangle,\langle b,\nullvalue\rangle\}$.

To explain and characterise the different behaviours of null values in databases, three distinct semantics have been proposed in the literature -- see, e.g., the seminal paper by \citet{Zaniolo:84}. 

\begin{enumerate}[a)]
	\item In the \emph{unknown} (also called \emph{existential}) semantics for null values, a null value denotes an existing but unknown value -- for example, in a \texttt{person} table the tuples corresponding to married people for which the identity of the spouse is unknown, should have an \emph{unknown} null value in the column \texttt{spouse}. The same \emph{unknown} semantics is given to \emph{Codd tables} \citep{Imielinski:1984} and to \emph{RDF blank nodes} \citep{franconi:et:al:rdf:05}.
	\item In the \emph{nonexistent} (also called \emph{inapplicable}) semantics for null values, the presence of a null value states the fact that no value exists whatsoever or, in other words, that the property represented by a column is \emph{inapplicable} for the tuple in which the null value appears. For example, in a \texttt{person} table the tuples corresponding to unmarried people should have a \emph{nonexistent} null value in the column \texttt{spouse}. A formal account of inapplicable nulls has been given by \citet{lerat:lipski:86}.
	\item In the \emph{no information} semantics for null values, null values are interpreted as having either inapplicable semantics or unknown semantics -- namely, it is not specified which of the two meanings can be associated to the null value. A formal account of no information semantics for null values has been given by, e.g., \citet{kohler_possible_2016}.
\end{enumerate}

Let's explore now the different behaviour of queries under the three different semantics. Under the \emph{unknown} semantics of null values the answer to the query (1) above would be different from the SQL answer, and indeed it would behave as the identity query, since the null value states the existence of a value without stating exactly which one:

\vspace{1ex}\noindent
\begin{minipage}{.075\textwidth}
~
\end{minipage}
\begin{minipage}{.6\textwidth}
{\footnotesize
\begin{verbatim}
SELECT c1, c2 FROM r WHERE c1 = c1 AND c2 = c2;
\end{verbatim}
}
\end{minipage}
\begin{minipage}{.3\textwidth}
\hfill
{\footnotesize
	$\stackanchor{\text{\tiny{unknown}}}{\text{\tiny{semantics}}}\!\dashrightarrow\  
	\begin{array}{|@{\hspace{1ex}}wc{1em}@{\hspace{1ex}}|@{\hspace{1ex}}wc{1em}@{\hspace{1ex}}|}
	\hline
		\texttt{c1}          & \texttt{c2}          \\ 
		\hline
		\texttt{a}          & \texttt{a}          \\ 
		\texttt{b}          & \nullvalue          \\ 
		\hline
	\end{array}$
}
\end{minipage}

\vspace{1ex}
\noindent
Let's try now a SQL query asking for all the tuples having a null value in column \texttt{c1}: 

\vspace{1ex}\noindent
\begin{minipage}{.075\textwidth}
\hspace{1em}(2)
\end{minipage}
\begin{minipage}{.6\textwidth}
{\footnotesize
\begin{verbatim}
SELECT c1, c2 FROM r WHERE c1 = NULL;
\end{verbatim}
}
\end{minipage}
\begin{minipage}{.3\textwidth}
\hfill
{\footnotesize
	$\stackanchor{\text{\tiny{SQL}}}{\text{\tiny{semantics}}}\!\dashrightarrow\  
	\begin{tabular}{|@{\hspace{1ex}}wc{1em}@{\hspace{1ex}}|@{\hspace{1ex}}wc{1em}@{\hspace{1ex}}|}
	\hline
		\texttt{c1}          & \texttt{c2}          \\ 
		\hline
		~       &   ~       \\
		\hline
	\end{tabular}$
}
\end{minipage}

\vspace{1ex} 
\noindent
 (note that we would get the same result by replacing the \texttt{WHERE} clause with ``\texttt{WHERE c1 IS NULL}''). This query yields in SQL the empty answer. However, if the null value in the query is interpreted with \emph{unknown} semantics -- meaning that there is a value but I don't know exactly which one -- then the above query (2) should behave as the identity query. Indeed, the \emph{unknown} null value in the query can be understood as matching any of the active domain values:

\vspace{1ex}\noindent
\begin{minipage}{.075\textwidth}
~
\end{minipage}
\begin{minipage}{.6\textwidth}
{\footnotesize
\begin{verbatim}
SELECT c1, c2 FROM r WHERE c1 = NULL;
\end{verbatim}
}
\end{minipage}
\begin{minipage}{.3\textwidth}
\hfill
{\footnotesize
	$\stackanchor{\text{\tiny{unknown}}}{\text{\tiny{semantics}}}\!\dashrightarrow\  
	\begin{array}{|@{\hspace{1ex}}wc{1em}@{\hspace{1ex}}|@{\hspace{1ex}}wc{1em}@{\hspace{1ex}}|}
	\hline
		\texttt{c1}          & \texttt{c2}          \\ 
		\hline
		\texttt{a}          & \texttt{a}          \\ 
		\texttt{b}          & \nullvalue          \\ 
		\hline
	\end{array}$
}
\end{minipage}

\vspace{1ex}
\noindent We can conclude from the behaviour of SQL queries (1) and (2) that SQL null values do not have an \emph{unknown} semantics.

\vspace{1ex}
On the other hand, the \emph{inapplicable} semantics for null values would informally explain the behaviour of the SQL queries (1) and (2): if a null value points out a non existing value, then a join of a null value with any actual value or a non existing value would fail, and a null value in a query would not match any actual value or a non existing value. Indeed, queries (1) and (2) under \emph{inapplicable} semantics would return the same answer as in SQL:

\vspace{1ex}\noindent
\begin{minipage}{.075\textwidth}
~
\end{minipage}
\begin{minipage}{.6\textwidth}
{\footnotesize
\begin{verbatim}
SELECT c1, c2 FROM r WHERE c1 = c1 AND c2 = c2;
\end{verbatim}
}
\end{minipage}
\begin{minipage}{.3\textwidth}
\hfill 
{\footnotesize
	$\stackanchor{\text{\tiny{inapplicable}}}{\text{\tiny{semantics}}}\!\dashrightarrow\  
	\begin{array}{|@{\hspace{1ex}}wc{1em}@{\hspace{1ex}}|@{\hspace{1ex}}wc{1em}@{\hspace{1ex}}|}
	\hline
		\texttt{c1}          & \texttt{c2}          \\ 
		\hline
		\texttt{a}          & \texttt{a}          \\ 
		\hline
	\end{array}$
}
\end{minipage}

\vspace{1ex}\noindent
\begin{minipage}{.075\textwidth}
~
\end{minipage}
\begin{minipage}{.6\textwidth}
{\footnotesize
\begin{verbatim}
SELECT c1, c2 FROM r WHERE c1 = NULL;
\end{verbatim}
}
\end{minipage}
\begin{minipage}{.3\textwidth}
\hfill
{\footnotesize
	$\stackanchor{\text{\tiny{inapplicable}}}{\text{\tiny{semantics}}}\!\dashrightarrow\  
	\begin{array}{|@{\hspace{1ex}}wc{1em}@{\hspace{1ex}}|@{\hspace{1ex}}wc{1em}@{\hspace{1ex}}|}
	\hline
		\texttt{c1}          & \texttt{c2}          \\ 
		\hline
		~          & ~         \\ 
		\hline
	\end{array}$
}
\end{minipage}

\vspace{1ex}
\noindent
But let's try now the SQL query (3), asking for the \texttt{c2} column of the table \texttt{r}: 

\vspace{1ex}\noindent
\begin{minipage}{.075\textwidth}
\hspace{1em}(3)
\end{minipage}
\begin{minipage}{.6\textwidth}
{\footnotesize
\begin{verbatim}
SELECT c2 FROM r;
\end{verbatim}
}
\end{minipage}
\begin{minipage}{.3\textwidth}
\hfill
{\footnotesize
	$\stackanchor{\text{\tiny{SQL}}}{\text{\tiny{semantics}}}\!\dashrightarrow\  
	\begin{array}{|@{\hspace{1ex}}wc{1em}@{\hspace{1ex}}|}
	\hline
		\texttt{c2}         \\ 
		\hline
		\texttt{a}         \\ 
		\nullvalue         \\ 
		\hline
	\end{array}$
}
\end{minipage}

\vspace{1ex}
\noindent
According to the \emph{inapplicable} semantics for null values, the first tuple in the answer is strictly \emph{more informative} than the second one~\citep{lien:82,lerat:lipski:86,roth_null_1989}. The first tuple is said to \emph{subsume} by the second tuple. In other words, there can not be a tuple with an inapplicable null value if in the same table there is an equal tuple with an actual value instead of the  null value.
As a matter of fact, according to the inapplicable semantics for null values, a minimality principle based on subsumption applies to the outcome of each algebraic operator~\citep{lerat:lipski:86}, and the result of query (3) should be minimised as follows:

\vspace{1ex}\noindent
\begin{minipage}{.075\textwidth}
~
\end{minipage}
\begin{minipage}{.6\textwidth}
{\footnotesize
\begin{verbatim}
SELECT c2 FROM r;
\end{verbatim}
}
\end{minipage}
\begin{minipage}{.3\textwidth}
\hfill
{\footnotesize
	$\stackanchor{\text{\tiny{inapplicable}}}{\text{\tiny{semantics}}}\!\dashrightarrow\  
	\begin{array}{|@{\hspace{1ex}}wc{1em}@{\hspace{1ex}}|}
	\hline
		\texttt{c2}         \\ 
		\hline
		\texttt{a}         \\ 
		\hline
	\end{array}$
}
\end{minipage}

\vspace{1ex}
\noindent
The peculiar semantics of null values with inapplicable semantics influences also how constraints are validated over databases with null values, e.g., as discussed by \citet{atzeni_functional_1986} and \citet{hartmann_implication_2012}.\\
Moreover, if we consider the SQL query (4) below, asking for the \texttt{c2} column of the table \texttt{r} restricted to tuples with a non-null value in column \texttt{c2}:

\vspace{1ex}\noindent
\begin{minipage}{.075\textwidth}
\hspace{1em}(4)
\end{minipage}
\begin{minipage}{.6\textwidth}
{\footnotesize
\begin{verbatim}
SELECT c2 FROM r WHERE c2 = c2;
\end{verbatim}
}
\end{minipage}
\begin{minipage}{.3\textwidth}
\hfill
{\footnotesize
	$\stackanchor{\text{\tiny{SQL}}}{\text{\tiny{semantics}}}\!\dashrightarrow\  
	\begin{array}{|@{\hspace{1ex}}wc{1em}@{\hspace{1ex}}|}
	\hline
		\texttt{c2}         \\ 
		\hline
		\texttt{a}         \\ 
		\hline
	\end{array}$
}
\end{minipage}

\vspace{1ex}
\noindent
we observe that, for any instance of the table \texttt{r} containing at least an actual non-null value in column \texttt{c2}, the query (3) and the query (4) are equivalent under the \emph{inapplicable} semantics, namely, query (3) and query (4) return the same answer. As we have seen, this is clearly not true with the SQL semantics for null values.

\noindent
We can conclude from the behaviour of SQL queries (3) and (4) that SQL null values do not have an \emph{inapplicable} semantics.

\vspace{1ex}
In the \emph{no information} semantics a null value can be understood as matching either any of the active domain values or a non existing value. It is easy to see that under no information semantics the query (2) we have seen before, asking for all the tuples having a null value in column \texttt{c1}, should behave, unlike SQL, as the identity query:

\vspace{1ex}\noindent
\begin{minipage}{.075\textwidth}
~
\end{minipage}
\begin{minipage}{.6\textwidth}
{\footnotesize
\begin{verbatim}
SELECT c1, c2 FROM r WHERE c1 = NULL;
\end{verbatim}
}
\end{minipage}
\begin{minipage}{.3\textwidth}
\hfill
{\footnotesize
	$\stackanchor{\text{\tiny{no information}}}{\text{\tiny{semantics}}}\!\dashrightarrow\  
	\begin{array}{|@{\hspace{1ex}}wc{1em}@{\hspace{1ex}}|@{\hspace{1ex}}wc{1em}@{\hspace{1ex}}|}
	\hline
		\texttt{c1}          & \texttt{c2}          \\ 
		\hline
		\texttt{a}          & \texttt{a}          \\ 
		\texttt{b}          & \nullvalue          \\ 
		\hline
	\end{array}$
}
\end{minipage}

\vspace{1ex}
\noindent
We can conclude from the behaviour of query (2) that SQL null values do not have a \emph{no information} semantics. 

\medskip

All these examples show that SQL null values do not have any of the above-mentioned three semantics for null values.

Informally, the above behaviour of SQL with null values could be explained by the fact that an SQL null value is never equal (or not equal) to anything, including other null values or even itself. How can we formally capture this behaviour in an algebra extending the standard relational algebra, and prove that it corresponds exactly to the SQL behaviour? How can we find an adequate model theory and a domain calculus for it?

In order to specify the semantics of SQL with null values, we first define a relational algebra -- called \nRA, see~\citep{franconi:tessaris:null:a:12,franconi:tessaris:null:b:12} -- specifying directly the \emph{behaviour} of SQL with null values. In \nRA, the null value is an actual value, and it is part of the domain of interpretation of the algebra. Therefore, a null value, as any other value, may appear explicitly in a database representation within the algebra, just like null values appear explicitly in SQL databases. We prove the equivalence of the fragment of \nRA without zero-ary relations with the \fosql first-order fragment of SQL, in which we do not consider aggregates and bags -- see Figure~\ref{fig:fosql}. \nRA is a simple extension of the standard relational algebra, which is well known to capture \fosql without null values. A similar result, extended with bag semantics, has been found recently by~\citep{guagliardo:libkin:17,libkin_peterfreund_2020}, albeit with a restriction which disallows the introduction of constants in the language and without zero-ary relations. 

A relational calculus, instead of specifying directly the behaviour of SQL, relies on a \emph{declarative} logic-based approach in order to specify the data of interest, without specifying how to obtain it. When looking for a relational calculus equivalent to \nRA, we have imposed the requirement that its model theory should not include the null value in its domain of interpretation: this is exactly what it is expected since the null value is just a syntactic marker representing a \textit{missing} value and it should not correspond to some special ``null'' element in the data. It would be meaningless for the data in the domain of interest to include an actual value labelled ``null'', since this would not identify any entity in the domain. Depending on the semantics given to null values (e.g., unknown or inapplicable), a null syntactic marker in the language would correspond to some actual value in the data or to the fact that there is not any value at all for the property.

The main result of our work is in the definition and analysis of two domain relational calculi, both proved to be equivalent to \nRA, satisfying our requirement above. The first one, called \FOLe, is just like the standard domain relational calculus with an additional syntactic term in the language representing the null marker, and where tuples are interpreted as \emph{partial} function from attributes to values, as opposed to \emph{total} functions like in standard domain relational calculus. The second one, the \textit{Decomposed Relational Calculus with Null Values}, is just like the standard domain relational calculus with an extended relational signature including the \emph{decomposition} of all its predicate symbols: the signature includes, for each predicate symbol $R$ of arity $n$ and for each subset of positional attributes $A\subseteq \{1,\cdots , n\}$, the predicate $\subR[R]_A$ of arity $|A|$. The decomposed signature corresponds to the \emph{horizontal decomposition} of each relation \citep{date:darwen:2010}, so that the language does not need to have an explicit marker for the null value.

The elegance and the simplicity of these two equivalent domain relational calculi -- both equivalent to \nRA and therefore able to capture \fosql -- may provide two alternative intuitive explanations of SQL null values in \fosql\!\!. One could see a SQL table either as a set of \emph{partial tuples}, or as a set of \emph{horizontally decomposed tables} each one holding regular total tuples. In the former case, we need a null marker to represent the partiality of a tuple, in the latter case decomposed tables hold partial tuples explicitly without using a null marker. One may argue that the presence of a SQL null value in a tuple could be an indication of bad modelling of the domain: according to this view, a good model should use an explicit decomposed table for that tuple without the null column, capturing the fact that the missing column is inapplicable for the tuple.

\medskip
The paper is organised as follows. We first introduce the \fosql first order fragment of SQL, and we show how it can be reduced to a non-standard variant of SQL  with classical two-valued logic semantics instead of the standard SQL three-valued logic semantics. Then, in Section~\ref{sec:preliminaries}, we discuss three isomorphic different representations of databases with null values, which will be used as the interpretation structures of the algebra and calculi developed later in the paper. In Section~\ref{sec:null_relational_algebra} we introduce the relational algebra with null values \nRA, as an obvious extension to the standard relational algebra \RA, and we show the expressive equivalence to \fosql of the fragment of \nRA without zero-ary relations. In Section~\ref{sec:null_relational_calculus} we introduce an extension to standard first-order logic \FOL, called \FOLe, which takes into account the possibility that some of the arguments of a relation might not exist. We then show that \FOLe is equally expressive in a strong sense to the standard first-order logic with a decomposed signature (Section~\ref{sub:logical_reconstruction}), and in Section~\ref{sec:queries} we study some of formal properties of the introduced calculi. In Section~\ref{sub:nra_and_nRC} we show our main result: \nRA is equally expressive to the domain independent fragment of \FOLe with the standard name assumption, a result which parallels in an elegant way the classical equivalence result by Codd for the standard relational algebra \RA and the  domain independent fragment of first-order logic \FOL with the standard name assumption. We conclude the paper with examples of usage of \fosql to represent standard constraints from SQL:99
.

\medskip
This paper is a major revision and extension of our previous papers~\citep{franconi:tessaris:null:a:12,franconi:tessaris:null:b:12}. Sections~\ref{sec:fosql},~\ref{sec:queries}, and~\ref{sec:conclusions} are completely new, Section~\ref{sec:intro}, \ref{sec:null_relational_algebra}, and~\ref{sec:null_relational_calculus} have been rewritten, few imprecisions have been fixed, and detailed proofs have been added.


\section{The first-order fragment of SQL \fosql}\label{sec:fosql}

In order to deal with null values in the relational model, \citet{codd:tods-79} included the special \texttt{NULL} value in the domain and adopted a three-valued logic having a third truth value \emph{unknown} together with \emph{true} and \emph{false}. The comparison expressions in Codd's algebra are evaluated to \emph{unknown} if they involve a null value, while in set operations, tuples otherwise identical but containing null values are considered to be different. SQL uses this three-valued logic for the evaluation of \texttt{WHERE} clauses -- see, e.g., \citep{Date:2003}. 

In this Section, we study the \fosql fragment of \sql, as defined in Figure~\ref{fig:fosql}. \fosql is a simple bags-free (enforced by ``\texttt{SELECT DISTINCT}'') and aggregation-free fragment of SQL.

\lstset{basicstyle=\ttfamily\footnotesize}
\newsavebox{\fosqlbox}
\begin{lrbox}{\fosqlbox}
\begin{lstlisting}
create-table := CREATE TABLE table 
                 ( { column { TEXT | INTEGER } [ NOT ] NULL , }+
                   { CONSTRAINT constraint-name 
                                CHECK (simple-condition) , }* ;

query := SELECT DISTINCT { expression [ AS column ] }+ 
         FROM { table [ AS table ] | (query) AS table }+ 
         WHERE condition ;
         | query { UNION | INTERSECT | EXCEPT } query ;

simple-condition := simple-comparison  
                    | (simple-condition) 
                    | NOT simple-condition 
                    | simple-condition { AND | OR } simple-condition

condition := atomic-condition 
             | (condition) 
             | NOT condition 
             | condition { AND | OR } condition

atomic-condition := simple-comparison
                    | [ NOT ] EXISTS query
                    | expression-tuple [ NOT ] IN 
                          ({ query | expression-tuple })
                  
simple-comparison := expression-tuple { = | != } expression-tuple
                     | expression IS [ NOT ] NULL

expression-tuple := expression | ({ expression }+)

expression := table.column | column | value

value := 'text' | integer | NULL
\end{lstlisting}
\end{lrbox}

\begin{figure}
{\footnotesize
  \fbox{\begin{minipage}{\textwidth}\usebox{\fosqlbox}\end{minipage}}
}
\caption{The syntax of \fosql}\label{fig:fosql}
\end{figure}


We show how any \fosql query can be transformed into an equivalent query with classical two-valued logic semantics instead of the standard SQL three-valued logic semantics, with the null value being just an additional domain value, and simple comparisons behaving in a special way. This transformation will be crucial to identify the algebra with null values characterising \fosql.
This result, which we first stated in~\citep{franconi:tessaris:null:a:12,franconi:tessaris:null:b:12} for the fragment of \fosql with the ``\texttt{WHERE}'' condition  composed only by simple comparisons, has been recently shown in a different way in~\citep{guagliardo:libkin:17,ricciotti_formalization_2020} for slightly different fragments of \fosql but generalised with bag semantics. We show in this Section our alternative proof.

The impact of null values in \fosql is only in the evaluation of the truth value of the condition in the ``\texttt{WHERE}'' clause, which is evaluated under three-valued (\textit{true}, \textit{false}, and \textit{unknown}) logic semantics (see Figure~\ref{fig:3vl}).
In absence of null values, the condition in the ``\texttt{WHERE}'' clause never evaluates to the unknown value, and the three-valued logic (3VL) semantics coincides with the classical two-valued logic (2VL) semantics. That's why we can reduce null-free \fosql into standard relational algebra and calculus.
Note that the constraint ``\texttt{CHECK (simple-condition)}'' on a table \texttt{r} is satisﬁed, by definition (see \citep{TuGe:VLDB:01a}), if and only if the condition ``\texttt{NOT EXISTS (SELECT 'fail' FROM r WHERE NOT (simple-condition))}'' yields \textit{true}, and so also in this case the impact of null values is only in the evaluation of the truth value of the condition in the ``\texttt{WHERE}'' clause.

\begin{figure}
\begin{center}
\includegraphics[scale=0.2]{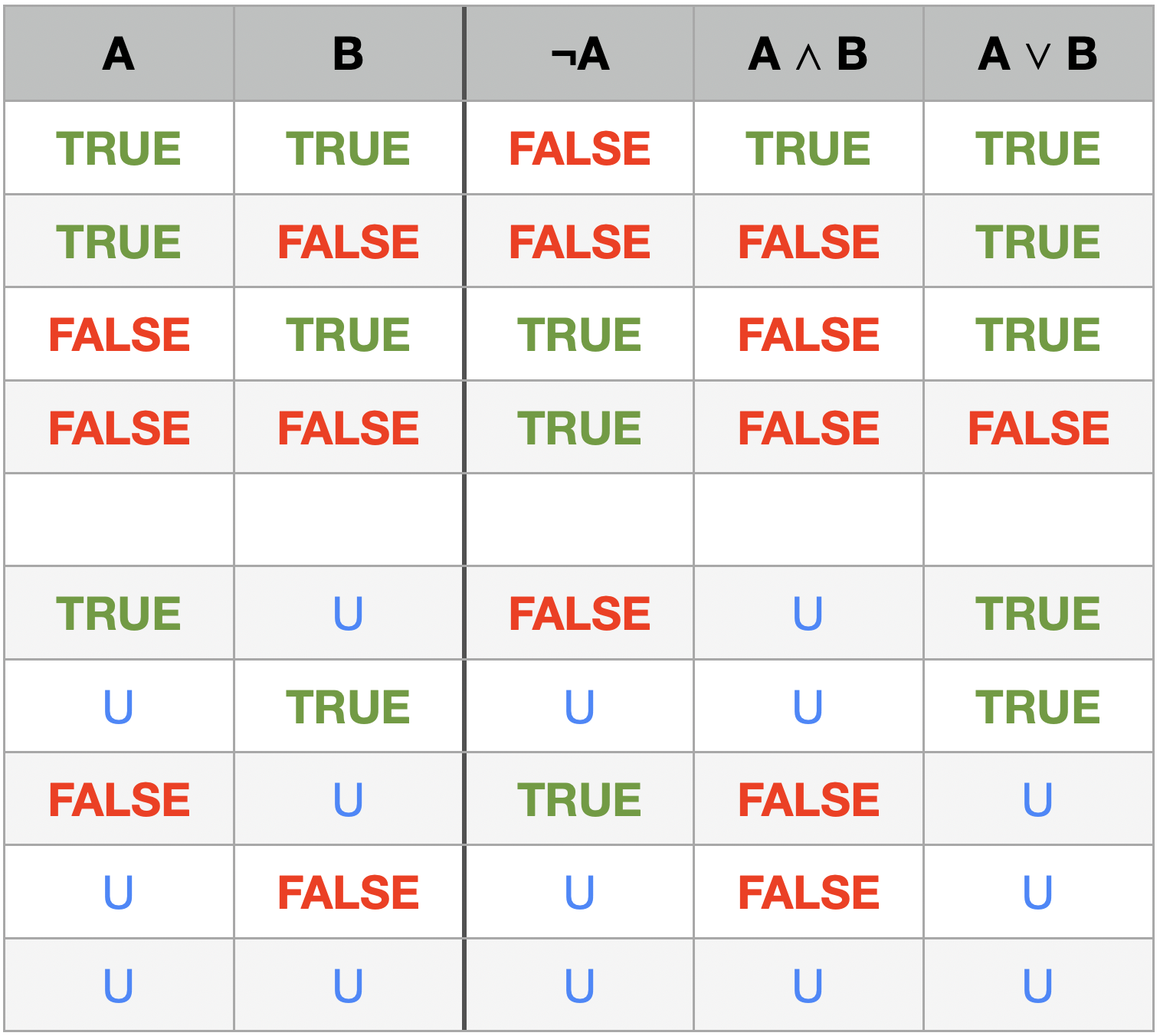}	
\end{center}
\caption{The three-valued (\textit{true}, \textit{false}, and \textit{unknown}) logic semantics}\label{fig:3vl}\end{figure}

Now, we are going to focus on the evaluation of the condition in the ``\texttt{WHERE}'' clause. 

\begin{definition} [Equality ``='' under 3VL semantics]
Equality ``='' between two values evaluates under 3VL semantics to:
\begin{itemize}
	\item \textit{true} if the two values are equal and each is different from \texttt{NULL};  
	\item \textit{false} if the two values are not equal and each is different from \texttt{NULL};
	\item \textit{unknown} if at least one value is \texttt{NULL}.\qed
\end{itemize}
Equality ``='' between two tuples of values of the same arity evaluates under 3VL semantics to:
\begin{itemize}
	\item \textit{true} if the 3VL equality ``='' between each element of the left tuple and the corresponding element of the right tuple evaluates to \textit{true};
	\item \textit{false} if the 3VL equality ``='' between some element of the left tuple and the corresponding element of the right tuple evaluates to \textit{false};
	\item \textit{unknown} otherwise.\qed
\end{itemize}
The 3VL semantics negation of the ``\texttt{=}'' simple comparison is the ``\texttt{!=}'' simple comparison. \qed
\end{definition} 

\begin{definition} [``\texttt{IS NULL}'' simple comparison]
The ``\texttt{IS NULL}'' unary simple comparison evaluates under 3VL semantics to:
\begin{itemize}
	\item \textit{true} if the value is \texttt{NULL};
	\item \textit{false} otherwise.
\end{itemize}
The 3VL semantics negation of the ``\texttt{IS NULL}'' simple comparison is the ``\texttt{IS NOT NULL}'' unary simple comparison. \qed
\end{definition} 

\begin{definition} [``\texttt{EXISTS}'' atomic condition]
The ``\texttt{EXISTS}'' atomic condition evaluates under 3VL semantics to:
\begin{itemize}
	\item \textit{true} if the query returns a non empty set of tuples;
	\item \textit{false} otherwise.
\end{itemize}
The 3VL semantics negation of the ``\texttt{EXISTS}'' atomic condition is the ``\texttt{NOT EXISTS}'' atomic condition. \qed
\end{definition} 

\begin{definition} [``\texttt{IN}'' atomic condition]
The ``\texttt{IN}'' atomic condition evaluates under 3VL semantics to:
\begin{itemize}
	\item \textit{true} if the 3VL equality ``=''  between the left expression and some element of the right expression set or of the result set of the query evaluates to \textit{true};
	\item \textit{false} if the 3VL equality ``=''  between the left expression and each element of the right expression set or of the result set of the query evaluates to \textit{false};
	\item \textit{unknown} otherwise. 
\end{itemize}
The 3VL semantics negation of the ``\texttt{IN}'' atomic condition is the ``\texttt{NOT IN}'' atomic condition. \qed
\end{definition}

We observe that atomic conditions are closed under negation, so that an arbitrary level of negations in front of any atomic condition can be \textit{absorbed} in the usual way to obtain an atomic condition equivalent under 3VL semantics.

We now show that by normalising the condition in the ``\texttt{WHERE}'' clause, we can transform any \fosql query into an equivalent (under 3VL semantics) query where the condition does not contain any explicit negation.

\begin{lemma} [de Morgan with absorption]
de Morgan's laws are valid under 3VL semantics: it is always possible to normalise a Boolean combination of atomic conditions under 3VL semantics to an equivalent (under 3VL semantics) negation normal form (NNF), where (possibly multiple levels of) negations appear only in front of atomic conditions. Negations can be absorbed, in order to obtain an equivalent (under 3VL semantics) absorbed negation normal form (ANNF) as a negation-free combination of conjunctions and disjunctions of atomic conditions. The transformation into absorbed negation normal form increases only linearly the size of the original Boolean combination of atomic conditions.
\end{lemma}
\begin{proof}
Equivalence can be checked by using the 3VL truth tables:
\begin{center}\includegraphics[scale=0.2]{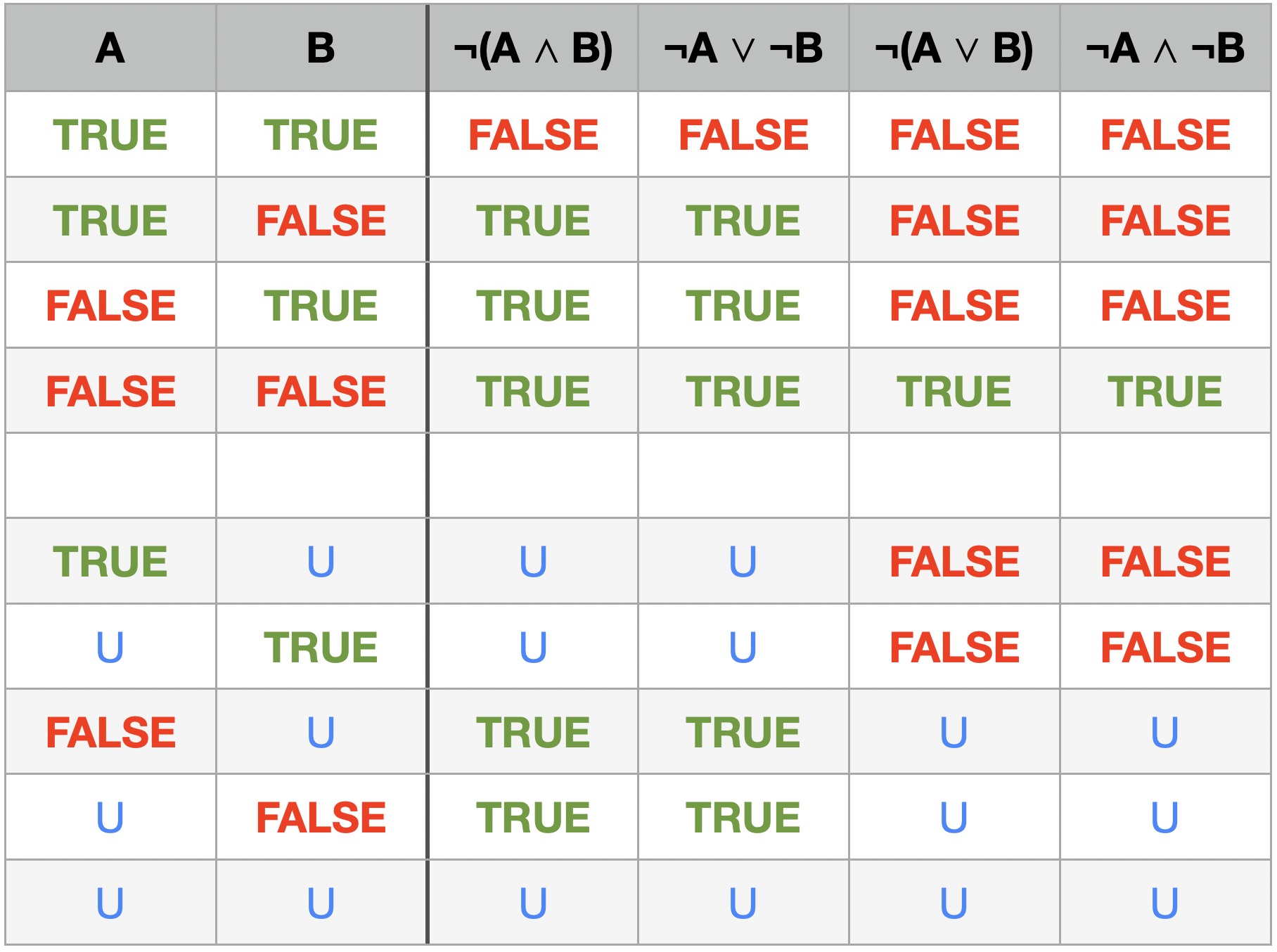}\qed\end{center}	
\end{proof} 

According to the \sql standard, the ``\texttt{SELECT}'' statement only considers rows for which the condition in the ``\texttt{WHERE}'' clause evaluates to \textit{true} under 3VL semantics. If the condition is in absorbed negation normal form, we show that we can evaluate the condition under classical 2VL semantics, with the proviso of evaluating each atomic condition to \textit{false} if it evaluates to \textit{unknown} under 3VL semantics -- we call these atomic conditions \textit{atomic conditions under classical 2VL semantics}.

\begin{definition} [Atomic conditions under classical 2VL semantics]~\\
	A (possibly absorbed) atomic condition under classical 2VL semantics evaluates to \textit{true} if and only if it evaluates to \textit{true} under 3VL semantics, and it evaluates to \textit{false} otherwise:
\begin{center}
\includegraphics[scale=0.2]{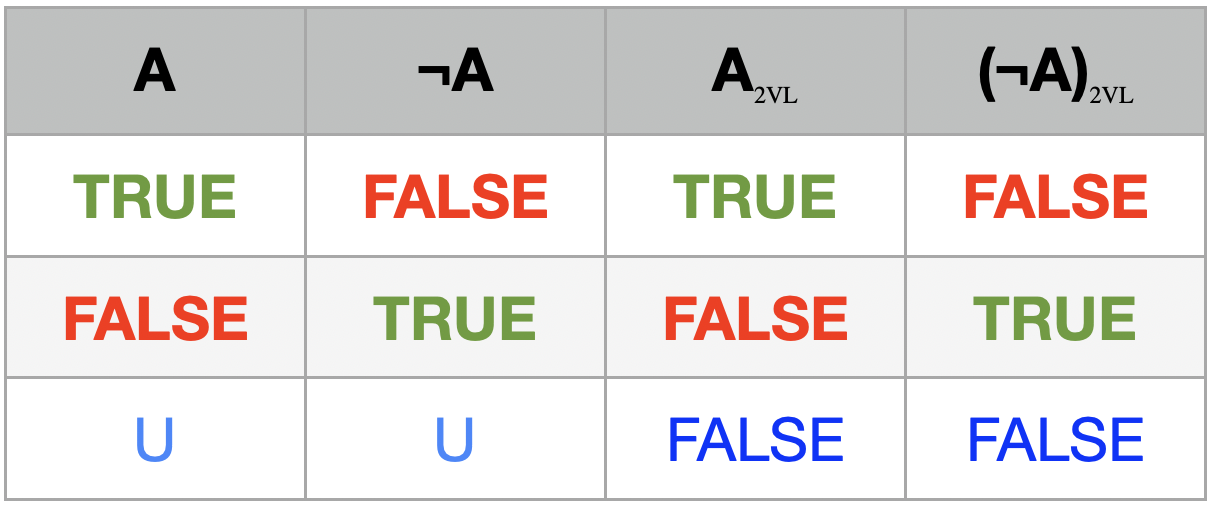}\hfill~
\end{center}	
Note that \textbf{\textsf{\small ($\neg$A)}}$_{\textrm{\tiny 2VL}}$ denotes the \textit{absorbed} negated atomic condition \textbf{\textsf{\small A}} under classical 2VL semantics.\qed
\end{definition}

\begin{lemma} [Equivalence of ``\texttt{WHERE}'' clause]
A Boolean combination of atomic conditions evaluates to \emph{true} under 3VL semantics
if and only if
the absorbed negation normal form of the Boolean combination of atomic conditions under classical 2VL semantics evaluates to \emph{true}.
\end{lemma}
\begin{proof}
Equivalence can be checked by comparing the \textit{true} values in the truth table below on the left (corresponding to the definition of 3VL semantics of Figure~\ref{fig:3vl}) with the the \textit{true} values in the truth table on the right, representing the truth values of absorbed atomic conditions and of their Boolean combinations under classical 2VL semantics. 
\begin{center}
\includegraphics[scale=0.2]{3vl.png}~~~
\includegraphics[scale=0.2]{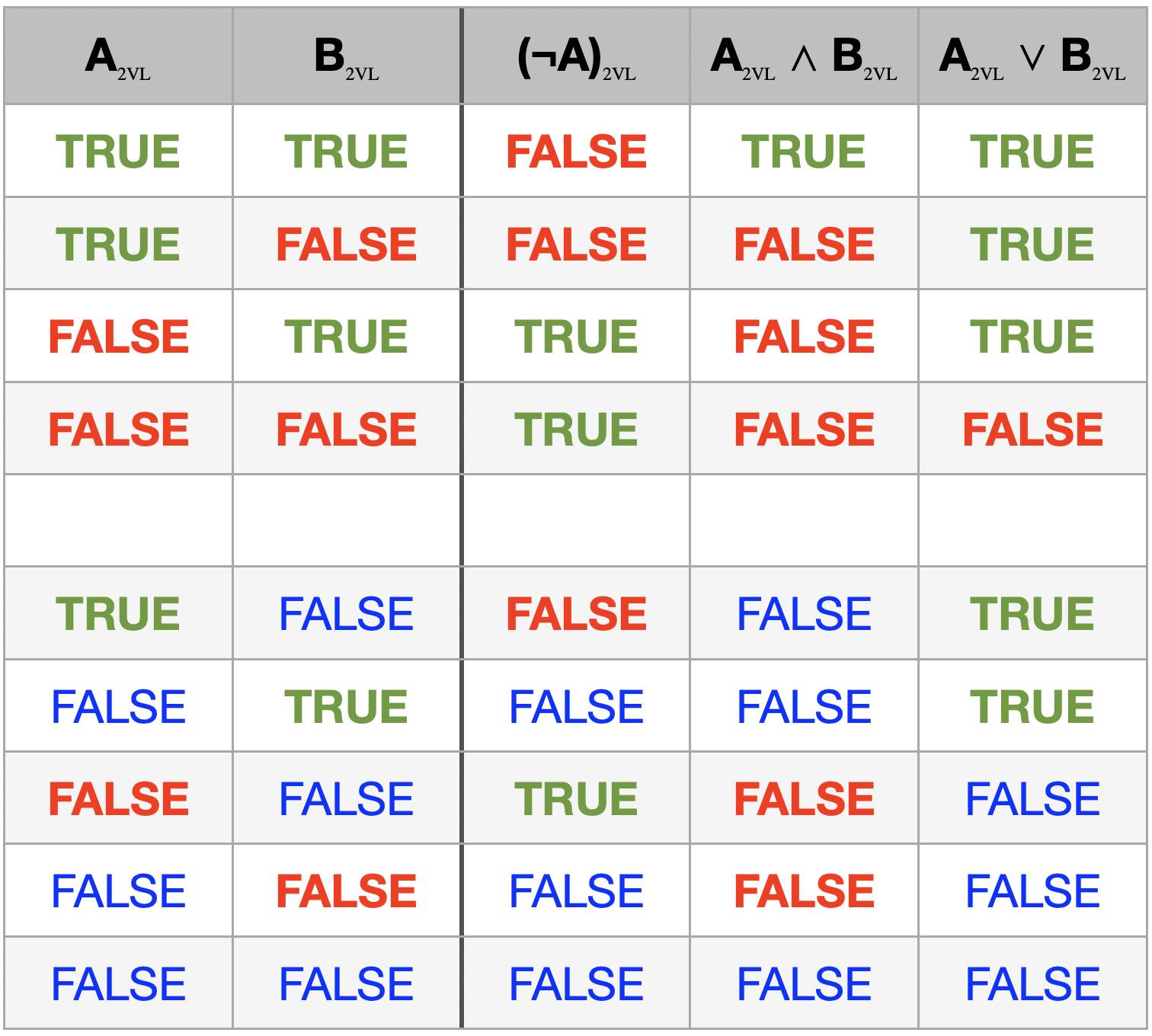}\qed
\end{center}	
\end{proof}

Let's see how we can encode atomic conditions under classical 2VL semantics, by just relying on the equality predicate ``\texttt{=$_{\textrm{\tiny 2VL}}$}\!'' and its negation ``\texttt{!=$_{\textrm{\tiny 2VL}}$}\!'' under classical 2VL semantics. Note that both ``\texttt{=$_{\textrm{\tiny 2VL}}$}\!'' and ``\texttt{!=$_{\textrm{\tiny 2VL}}$}\!'' behave just like standard equality and inequality but they fail when any of their arguments is the null value.
 
{\footnotesize
\begin{lstlisting}
e = f $\leadsto$
  e =$_{\textrm{\tiny 2VL}}$ f
\end{lstlisting}
\begin{lstlisting}
e != f $\leadsto$
  e !=$_{\textrm{\tiny 2VL}}$ f
\end{lstlisting}
\begin{lstlisting}
e IN (f$_1$ ... f$_n$) $\leadsto$
  e =$_{\textrm{\tiny 2VL}}$ f$_1$ OR ... OR e =$_{\textrm{\tiny 2VL}}$ f$_n$
\end{lstlisting}
\begin{lstlisting}
e NOT IN (f$_1$ ... f$_n$) $\leadsto$
  e !=$_{\textrm{\tiny 2VL}}$ f$_1$ AND ... AND e !=$_{\textrm{\tiny 2VL}}$ f$_n$
\end{lstlisting}
\begin{lstlisting}
e$_1$ ... e$_n$ IN (q) $\leadsto$
  EXISTS (SELECT DISTINCT v.1 ... v.n FROM q AS v WHERE 
           e$_1$ =$_{\textrm{\tiny 2VL}}$ v.1 AND ... AND e$_n$ =$_{\textrm{\tiny 2VL}}$ v.n)
\end{lstlisting}
\begin{lstlisting}
e$_1$ ... e$_n$ NOT IN (q)) $\leadsto$
  NOT EXISTS (SELECT DISTINCT v.1 ... v.n FROM q AS v WHERE 
               ((e$_1$ =$_{\textrm{\tiny 2VL}}$ v.1 OR e$_1$ IS NULL) OR v.1 IS NULL) 
               AND ... AND
               ((e$_n$ =$_{\textrm{\tiny 2VL}}$ v.n OR e$_n$ IS NULL) OR v.n IS NULL)) 
\end{lstlisting}}

The main theorem can now be stated, as a simple consequence of the two above lemmata and the provided encoding.

\begin{theorem} [2VL evaluation of \fosql] \label{th:2vl-fosql}
The evaluation of a \fosql query is the same as the \fosql query with the condition in the ``\!\texttt{WHERE}'' clause transformed to an absorbed negation normal form. The evaluation of the transformed \fosql query is the same as that query evaluated using 2VL semantics for Boolean operators in the ``\texttt{WHERE}'' clause, by considering the null value as any other domain value, and encoding the atomic conditions under classical 2VL semantics.
\end{theorem} 

 Note that the results in this Section can be easily generalised to include more simple comparisons, like ``\texttt{<}'', ``\texttt{>}'', ``\texttt{$\leq$}'', ``\texttt{$\geq$}''.

\begin{example}
	The \fosql query:
{\footnotesize
\begin{lstlisting}
SELECT DISTINCT person.name
FROM person
WHERE NOT (person.passport != 'Italian' AND
           person.cityofbirth NOT IN 
             (SELECT DISTINCT city.name
              FROM city
              WHERE city.country = 'Italy'));
\end{lstlisting}}
\noindent
is transformed into the following query under classical 2VL semantics, according to Theorem~\ref{th:2vl-fosql}:
{\footnotesize
\begin{lstlisting}
SELECT DISTINCT person.name
FROM person
WHERE person.passport =$_{\textrm{\tiny 2VL}}$ 'Italian' OR
      EXISTS (SELECT DISTINCT italiancity.name
              FROM (SELECT DISTINCT city.name
                    FROM city
                    WHERE city.country =$_{\textrm{\tiny 2VL}}$ 'Italy') AS italiancity
              WHERE person.cityofbirth =$_{\textrm{\tiny 2VL}}$ italiancity.name);                
\end{lstlisting}}

\noindent
The above transformed query, which does not conform to standard SQL due to the special ``\texttt{=$_{\textrm{\tiny 2VL}}$}\!'' equality predicate, can be transformed into a standard SQL query by rewriting the occurrences of the ``\texttt{=$_{\textrm{\tiny 2VL}}$}\!'' equality predicate with the 3VL semantics ``\texttt{=}'' SQL equality predicate restricted to 2VL semantics via a ``\texttt{NOT NULL}'' check. For example, the 2VL condition \texttt{person.cityofbirth =$_{\textrm{\tiny 2VL}}$ italiancity.name} would be transformed as:

{\footnotesize
\begin{lstlisting}[escapechar=§]
person.cityofbirth =$_{\textrm{\tiny 2VL}}$ italiancity.name $~~\leadsto~~$
\end{lstlisting}
\vspace{-2ex}
\begin{lstlisting}[escapechar=§]
  person.cityofbirth = italiancity.name AND  
  person.cityofbirth IS$\!\!$ NOT$\!\!$ NULL AND  
  italiancity.name IS$\!\!$ NOT$\!\!$ NULL  §\qed
\end{lstlisting} 
}	
\end{example}


\section{Database Instances with Null Values}
\label{sec:preliminaries}

In this section we discuss some necessary formal preliminaries in order to properly define an algebra and a calculus equivalent to \fosql. We introduce the notions of tuple, relation, and database instance in presence of null values. To simplify the formalism we consider the unnamed (positional) perspective for the attributes of tuples: elements of a tuple are identified by their position within the tuple. 

Given a set of domain values $\Delta$, an $n$-tuple over $\Delta$ is defined for $n\geq 0$: 
if $n\geq 1$ then an $n$-tuple $\langle d_1,\cdots, d_n \rangle$ is the \emph{total} $n$-ary function $\{1,\cdots, n\}\mapsto\Delta\cup\set{\nullvalue}$, represented also as $\{1\mapsto d_1,\cdots, n\mapsto  d_n\}$ with $d_1,\cdots, d_n\in\Delta\cup\set{\nullvalue}$, 
while if $n=0$ then the zero-tuple $\langle\rangle$ is the \emph{constant} zero-ary function, represented also as \set{}. In a $n$-tuple with $n\geq 1$, the integer $i$ from 1 up to $n$ denotes the $i$-th \emph{position} of the element $d_i$ within the tuple: for example, the 2-tuple $\langle \texttt{b},\nullvalue \rangle$ with $b$ in position 1 and \nullvalue in position 2 is represented as $\{1\mapsto \texttt{b},2\mapsto \nullvalue\}$. 
In the following we may represent the set $\set{i\in\mathbb{N}\mid 1 \leq i \leq n}$ also as $[1\cdots n]$. 
A relation of arity $n$ is a set of $n$-tuples; if we want to specify that $n$ is the arity of a given relation $R$, we write the relation as $R/n$. The arity of a relation is also called \emph{degree}. Note that a relation of arity zero is either empty or it includes only the zero-tuple \set{}.
A relational schema $\R$ includes a set of relation symbols with their arities, a set of constant symbols $\C$, and the special term \nullvalue. A database instance \InRA{\I} includes a domain $\Delta\supseteq\C$, and it associates to each relation symbol $R$ of arity $n$ from the relational schema $\R$ a set of $n$-tuples $\InRA{\I}(R)$, to each constant symbol $c$ in $\C$ a domain value in $\Delta$ equal to $c$, and it associates to the \nullvalue term the special null value $\nullvalue\in\Delta$. 
As an example of a database instance \InRA{\I} consider Figure~\ref{fig:db-inst}(a).
 Database instances in SQL are represented as as \InRA{\I} databases instances, with an explicit representation of the null value in the domain of interpretation.

In our work we consider two alternative representations of a database instance \InRA{\I} where null values do not appear explicitly in the domain: \IFOLe{\I} and \IRA{\I}. We will show that they are isomorphic to \InRA{\I}. The representations share the same constant symbols, domain values (without \nullvalue), and mappings from constant symbols to domain values. The differences are in the way null values are encoded within a tuple. 

Compared with a database instance \InRA{\I}, a corresponding database instance \IFOLe{\I} differs only in the way it represents $n$-tuples: an $n$-tuple is a \emph{partial} function from integers from 1 up to $n$ into the set of domain values - the function is undefined exactly for those positional arguments which are otherwise defined and mapped to a null value in \InRA{\I}. 
As an example of a database instance \IFOLe{\I} consider Figure~\ref{fig:db-inst}(b).

\begin{figure}[t]%
\fbox{%
\begin{minipage}{.98\textwidth}%
{%
\small%
\vspace{-1em}\begin{minipage}[b]{.31\textwidth}
\begin{center}
\begin{align*}
\texttt{r}/2:\ & \texttt{}
\begin{aligned}[t]
\{ &\{1\mapsto \texttt{a},2\mapsto \texttt{a}\}, \\
   & \{1\mapsto \texttt{b},2\mapsto \nullvalue\}\}
\\&\\&
\end{aligned}
\end{align*}
\\
	(a) instance \InRA{\I}
\end{center}
\end{minipage}
\begin{minipage}[b]{.31\textwidth}
\begin{center}
\begin{align*}
\texttt{r}/2 :\ &
\begin{aligned}[t]
\{&\{1\mapsto \texttt{a},2\mapsto \texttt{a}\}, \\
&\{1\mapsto \texttt{b}\}\}
\\&\\&
\end{aligned}
\end{align*}
\\
	(b) instance \IFOLe{\I}
\end{center}
\end{minipage}
\begin{minipage}[b]{.35\textwidth}
\begin{center}
\begin{align*}
\subR[\texttt{r}]_{\set{1,2}}/2 :\ & \{\{1\mapsto \texttt{a},2\mapsto \texttt{a}\}\}\\
\subR[\texttt{r}]_{\set{1}}/1 :\ & \{\{1\mapsto \texttt{b}\}\}\\ 
\subR[\texttt{r}]_{\set{2}}/1 :\ & \{\}\\
\subR[\texttt{r}]_{\set{}}/0 :\ & \{\}
\end{align*}
\\
	(c) instance \IRA{\I}
\end{center}
\end{minipage}
}
\end{minipage}}
\caption{The three representations of the example database instance $\I$}
\label{fig:db-inst}
\end{figure}

Compared with a database instance \InRA{\I}, a corresponding database instance \IRA{\I} differs in the way relation symbols are interpreted: a relation symbol of arity $n$ is associated to a set of null-free total tuples of dishomogeneous arities up to $n$.
Given a database instance \InRA{\I} defined over a relational schema $\R$,
the corresponding database instance \IRA{\I} is defined over the \emph{decomposed} relational schema $\subR$: for each relation symbol $R\in\R$ of arity $n$ and for each (possibly empty) subset of its positional arguments $A\subseteq [1\cdots n]$, the decomposed relational schema $\subR$ includes a predicate $\subR[R]_A$ with arity $|A|$.
The correspondence between \InRA{\I} and \IRA{\I} is based on the fact that each $|A|$-tuple in the relation $\IRA{\I}(\subR[R]_A)$ corresponds exactly to the $n$-tuple in $\InRA{\I}(R)$ having non-null values only in the positional arguments in $A$, with the same values and in the same relative positional order. This corresponds to the notion of lossless \emph{horizontal decomposition} in relational databases, and it has been advocated by~\citet{date:darwen:2010} as the right way to represent null values.
As an example of a database instance \IRA{\I} consider Figure~\ref{fig:db-inst}(c).

In absence of null values, the \InRA{\I} and \IFOLe{\I} representations of a database instance coincide, and they coincide also with the \IRA{\I} representation if we equate each $n$-ary $R$ relation symbol in \InRA{\I} and \IFOLe{\I} with its corresponding $\subR[R]_{[1\cdots n]}$ relation symbol in \IRA{\I}. Indeed, for every $n$-ary relation $R$, in absence of null values, $\InRA{\I}(R)=\IFOLe{\I}(R)=\IRA{\I}(\subR[R]_{[1\cdots n]})$ and $\IRA{\I}(\subR[R]_A)=\emptyset$ for every relation of arity $|A|<n$.
Given the discussed isomorphisms, in the following -- whenever the difference in the representation of null values is not ambiguous -- we will denote as $\I$ the database instance represented in any of the above three forms \InRA{\I}, \IFOLe{\I}, or \IRA{\I}.

The different representations of a database instance $\I$ with null values are used in different contexts. The \InRA{\I} representation is the direct representation of database instances in SQL and in our proposed relational algebra with SQL nulls values \nRA. The \IFOLe{\I} representation is the representation of database instances in our proposed domain calculus with SQL null values. This representation preserves the SQL relational schema, and tuples are \textit{partial} functions to the domain of interpretation, which does not include the null value: this is exactly what it is expected since the null value represents a \textit{missing} value and it should not correspond to any element of the domain. The \IRA{\I} representation  is the representation of database instances in our proposed \textit{decomposed} domain calculus with SQL null values. This representation decomposes the SQL relational schema, and tuples are \textit{total} functions to the domain of interpretation, which does not include the null value.

\section{Relational Algebra with Null Values}
\label{sec:null_relational_algebra}


\begin{figure}[t]
\fbox{%
\begin{minipage}{\textwidth}
\noindent
\textbf{Atomic relation -} \graybox{${R}$} \textbf{-} (where $R\in\R$)

\hspace{2em}
$\eval{R} = \InRA{\I}(R).$

\vspace{1ex}\noindent
\textbf{Constant singleton -} \graybox{${\singleton{v}}$} \textbf{-} (where $v\in\C$) 

\hspace{2em}
$\eval{\singleton{v}} = \set{1\mapsto v}.$ 

\vspace{1ex}\noindent
\textbf{Selection -} \graybox{${\select{i = v}{e}}$},\ \graybox{${\select{i = j}{e}}$} \textbf{-} (where $v\in\C$,\ \ $\ell$ is the arity of $e$,\ \  and $1\leq i,j\leq\ell$)

\hspace{2em}
$	\eval{\select{i = v}{e}} = \set{ s\ \text{a $\ell$-tuple}  \mid s \in \eval{e} \land s(i) = v},$
 
\hspace{2em} 
$	\eval{\select{i = j}{e}} = \set{ s\ \text{a $\ell$-tuple}  \mid s \in \eval{e} \land s(i) = s(j)}.$

\vspace{1ex}\noindent
\textbf{Projection -} \graybox{${\project{i_1,\ldots,i_k}{e}}$} \textbf{-} (where $\ell$ is the arity of $e$,\ \  and $\set{i_1,\ldots,i_k} \subseteq [1\ldots\ell]$)

\hspace{2em}
$\eval{\project{i_1,\ldots,i_k}{e}} = \set{s\ \text{a $k$-tuple} \mid 	\begin{aligned}[t]
&\text{exists } s'\in \eval{e} \text{ s.t.} 
\text{  for all } 1\leq j\leq k .\  s(j)=s'(i_j)}.
\end{aligned}
$

\vspace{1ex}\noindent
\textbf{Cartesian product -} \graybox{${e \times e'}$} \textbf{-} (where $n,m$ are the arities of $e,e'$)

\hspace{2em} 
{\setlength\arraycolsep{0.5ex}
$\eval{e \times e'}  = \set{ s \text{ a }  
	\begin{aligned}[t]
	& (n+m)\text{-tuple}\mid \text{exists } t\in \eval{e},t'\in \eval{e'} \text{ s.t.} \\
	& 
	\begin{array}[t]{lrcll}
	\text{for all} & 1\leq &\! j\! &\leq n : & s(j)=t(j) \text{ and}\\
	\text{for all} & 1+n\leq &\! j\! &\leq (n+m) : & s(j)=t'(j-n)}. 
	\end{array}
\end{aligned}
$}

\vspace{1ex}\noindent
\textbf{Union/Difference -} \graybox{${e \cup e'}$},\ \graybox{${e - e'}$} \textbf{-} (where $\ell$ is the arity of $e$ and $e'$)

\hspace{2em}
$\eval{e \cup e'} = \{ s\ \text{a $\ell$-tuple} \mid s\in \eval{e} \vee s\in \eval{e'} \}$,

\hspace{2em} 
$\eval{e - e'} = \{ s\ \text{a $\ell$-tuple} \mid s\in \eval{e} \wedge s\not\in \eval{e'} \}$.

\vspace{1ex}\textbf{Derived operators -} where $v\in\C$,\ \ $\ell\leq\mathsf{min}(m,n)$,\ \ $m,n$ are the arities of $e,e'$,\\ $i,j,i_1,\ldots,i_\ell\leq m$,\ \ and $k_1,\ldots,k_\ell\leq n$

\hspace{2em} $
\begin{array}[t]{rclcrcl}
	e \cap e' &\doteq & e - (e - e'),\\
	\select{i \neq v}{e} &\doteq & e - \select{i = v}{e}, & &
	\select{i \neq j}{e} &\doteq& \select{i=i}{\select{j=j}{e}} - {\select{i = j}{e}},\\
	\select{c\land c'}{e} &\doteq &  \select{c}{e}\ \cap\select{c'}{e}, & &
	\select{c\lor c'}{e} &\doteq & \select{c}{e}\ \cup\select{c'}{e},\\
	\select{\neg(c\land c')}{e} &\doteq & \select{\neg c}{e}\ \cup\select{\neg c'}{e}, & &
	\select{\neg(c\lor c')}{e} &\doteq & \select{\neg c}{e}\ \cap\select{\neg c'}{e},\\
	e\RAJoin_{\substack{i_1=k_1\\\cdots\\i_\ell=k_\ell}} e' &\doteq &
	\multicolumn{5}{c}{\project{([1\cdots m+n]\setminus\set{m+k_1,\ldots,m+k_\ell})}{\select{i_1=m+k_1}{\ldots\select{i_\ell=m+k_\ell}{(e\times e').}}}}
\end{array}
$
\end{minipage}
}
\caption{The standard relational algebra \RA}\vspace{-1ex}
\label{fig:ra}\label{def:RA}
\end{figure}

We introduce in this Section the formal semantics of the relational algebra \nRA dealing with null values, which we  prove characterising \fosql exactly.

Let's first recall the notation of the standard relational algebra \RA (see, e.g., \citep{AbHV95} for details). 
Standard relational algebra expressions over a relational schema $\R$ are built according to the inductive formation rules in the boxed expressions of Figure~\ref{fig:ra}. Also in Figure~\ref{fig:ra} the semantics of an algebra expression $e$ is inductively defined as the transformation of database instances $\I$ -- with the Standard Name Assumption -- to a set of tuples \eval{e}. Given a tuple $t$ and a \RA expression $e$, we call \emph{models} of the expression $e$ with the answer $t$ all the database instances $\I$ such that $t\in \eval{e}$.

We remind that a relational schema for \RA may include zero-ary relations, and that the projection operator may project over the empty set of positional attributes, leading to a zero-ary derived relation. It is well known that zero-ary relations can not be represented as SQL (zero-ary) tables, as ~\citet{date_database_2010} discuss at length. As a matter of fact, the well known result that \RA is equally expressive as null-free \fosql \citep{ceri:gottlob:ieee-1985,paredaens_structure_1989,dadashzadeh_converting_1990,VandenBussche:et:al:2009} holds only for the case of \RA without zero-ary relations and with projections over non empty sets of positional attributes. We call \RA[0] the relational algebra \RA without zero-ary relations and with projections over non empty sets of positional attributes.

We extend the standard relational algebra to deal with SQL nulls; we refer to it as \emph{Relational Algebra with Null Values} (\nRA). In order to define \nRA from the standard relational algebra, we adopt -- just like SQL -- the \InRA{\I} representation of a database instance where null values are explicitly present as possible elements of tuples
. 
Figure~\ref{fig:nra} introduces the syntax and semantics of \nRA expressed in terms of its difference with \RA, including the derived operators involving null values which are added to \RA. All the \RA expressions are valid \nRA expressions, and maintain the same semantics, with the only change in the semantics of the \emph{selection} expressions $\mathbf{\select{i = j}{e}}$ and $\mathbf{\select{i = v}{e}}$: their semantic definitions make sure that the elements to be tested for equality are both different from the null value in order for the equality to succeed. In other words, it is enough to let equality fail whenever null values are involved. Equality in \nRA behaves exactly as equality ``\texttt{=}$_{\textrm{\tiny 2VL}}$''  under classical 2VL semantics in \fosql. In \nRA we will use the standard notation for equality ``\texttt{=}'' and inequality ``\texttt{$\neq$}'' with the assumption that they actually mean ``\texttt{=}$_{\textrm{\tiny 2VL}}$'' and ``\texttt{!=}$_{\textrm{\tiny 2VL}}$''.

\begin{figure}[t]
\fbox{%
\begin{minipage}{\textwidth}
\noindent
\textbf{Null singleton -} \graybox{${\singleton{\nullvalue}}$} 
\textbf{-}

\hspace{2em}
$\eval{\singleton{\nullvalue}} = \set{1\mapsto \nullvalue}.$

\vspace{1ex}\noindent
\textbf{Selection -} \graybox{${\select{i = v}{e}}$},\ \graybox{${\select{i = j}{e}}$} \textbf{-} (where $v\in\C$,\ \ $\ell$ is the arity of $e$,\ \  and $i,j\leq\ell$)

\hspace{2em}
$\eval{\select{i = v}{e}} = \set{ s\ \text{a $\ell$-tuple}  \mid s \in \eval{e} \land s(i) = v \land s(i) \neq \nullvalue},$

\hspace{2em} 
$\eval{\select{i = j}{e}} = \set{ s\ \text{a $\ell$-tuple}  \mid s \in \eval{e} \land s(i) = s(j) \land s(i) \neq \nullvalue \wedge s(j) \neq \nullvalue}.$

\vspace{1ex}\noindent
\textbf{Derived operators -} where $v\in\C$,\ \ $\ell\leq\mathsf{min}(m,n)$,\ \ $m,n$ are the arities of $e,e'$,\\ $i,j,i_1,\ldots,i_\ell\leq m$,\ \ and $k_1,\ldots,k_\ell\leq n$

\hspace{2em} $
\begin{array}[t]{rclrcl}
	\select{i = \nullvalue}{e} &\doteq& e - e, &
	\select{i \neq \nullvalue}{e} &\doteq& e - e, \\
	\select{\isnull(i)}{e} & \doteq& e - {\select{i = i}{e}}, &
	\select{\isnotnull(i)}{e} &\doteq& \select{i = i}{e},\\ 
	e \leftouterjoin_{\substack{i_1=k_1\\\cdots\\i_\ell=k_\ell}} e' &\doteq &
	\multicolumn{4}{c}{(e\RAJoin_{\substack{i_1=k_1\\\cdots\\i_\ell=k_\ell}}e')\cup(e - \project{1,\ldots,m}{(e\RAJoin_{\substack{i_1=k_1\\\cdots\\i_\ell=k_\ell}}e')})\times(\underbrace{\singleton{\nullvalue}\times\cdots\times\singleton{\nullvalue}}_{n-\ell}).}
\end{array}
$
\end{minipage}
}
\caption{The relational algebra with null values \nRA defined in the parts with null values different from \RA}\label{fig:nra}\label{def:nRA}
\end{figure}

\begin{example}\label{ex:nRA-identity}
The \textit{pseudo-identity} \nRA query
$(\select{1=1}{\select{2=2}{\texttt{r}}})$
evaluates to
$\{\langle \texttt{a},\texttt{a}\rangle\}$ over the database instance $\I$ with $\InRA{\I}(\texttt{r})=\{\langle \texttt{a},\texttt{a}\rangle,\langle \texttt{b},\nullvalue\rangle$\}, namely $\eval{\select{1=1}{\select{2=2}{\texttt{r}}}}=\{\langle \texttt{a},\texttt{a}\rangle\}$.\qed	
\end{example}

\begin{example}\label{ex:nRA-denial}
Consider the schema $\{\texttt{r}/2, \texttt{s}/2\}$ with the following database instance and \nRA denial constraints: 

\medskip
\noindent
{\small
\begin{minipage}{.15\textwidth}
	$\texttt{p}: 
	\begin{array}{|@{\hspace{1ex}}c@{\hspace{1ex}}|@{\hspace{1ex}}c@{\hspace{1ex}}|}
	\hline
		\texttt{1}          & \texttt{2}          \\ 
		\hline
		\texttt{a}          & \texttt{a}          \\ 
		\texttt{b} & \nullvalue          \\ 
		\nullvalue & \texttt{b}          \\ 
		\nullvalue & \nullvalue          \\ 
		\hline
	\end{array}$
\end{minipage}
\begin{minipage}{.14\textwidth}
	$\texttt{q}: 
	\begin{array}{|@{\hspace{1ex}}c@{\hspace{1ex}}|@{\hspace{1ex}}c@{\hspace{1ex}}|}
	\hline
		\texttt{1}          & \texttt{2}          \\ 
		\hline
		\texttt{a}          & \texttt{a}          \\ 
		\texttt{b}          & \nullvalue          \\ 
		\hline
	\end{array}$
\end{minipage}
}
\begin{minipage}{.8\textwidth}
\hspace{-1em}
\begin{itemize}

\item UNIQUE constraint for $\texttt{p.1}$:
$\select{1=3}{\select{2\neq 4}{(\texttt{p}\times \texttt{p})}}=\emptyset$

\item UNIQUE constraint for $\texttt{q.1}$:
$\select{1=3}{\select{2\neq 4}{(\texttt{q}\times \texttt{q})}}=\emptyset$

\item NOT-NULL constraint for $\texttt{q.1}$:
$\select{\isnull(1)}{\texttt{q}}=\emptyset$

\item FOREIGN KEY constraint from $\texttt{p.2}$ to $\texttt{q.1}$:\\
$\project{2}{\select{\isnotnull(2)}{\texttt{p}}} - \project{1}{\select{\isnotnull(1)}{\texttt{q}}}\ =\emptyset$

\end{itemize}

\end{minipage}
\medskip

\noindent
The constraints are all satisfied in the given database instance.\qed
\end{example}

We call \nRA[0] the relational algebra \nRA without zero-ary relations and with projections over non empty sets of positional attributes. As expected, \nRA[0] characterises exactly \fosql. 

\begin{theorem}
	\nRA[0] and \fosql are equally expressive.
\end{theorem}

\begin{proof}
We first observe that \nRA[0] is just like \RA[0] with an added constant $\nullvalue$, and where the standard equality predicate and its derived negation predicate in selections are replaced by predicates which behave exactly like the SQL equality predicate ``\texttt{=$_{\textrm{\tiny 2VL}}$}\!'' and its negation ``\texttt{!=$_{\textrm{\tiny 2VL}}$}\!'' under classical 2VL semantics.

\noindent
Any 2VL translated \fosql query (as defined in Theorem~\ref{th:2vl-fosql}) is a SQL query where the null value is like any other value, the conditions in the WHERE clause are composed only by classical 2VL semantics predicates (and therefore the Boolean operators reduce to just classical 2VL Boolean operators), and equality and its negation are the SQL equality predicate ``\texttt{=$_{\textrm{\tiny 2VL}}$}\!'' and its negation ``\texttt{!=$_{\textrm{\tiny 2VL}}$}\!'' under classical 2VL semantics.

\noindent
Using results about the translation of null-free SQL queries to relational algebra (as proposed, e.g., by \citet{ceri:gottlob:ieee-1985,dadashzadeh_converting_1990,VandenBussche:et:al:2009}), these queries, after having expanded the derived operators up to a negated normal form within selection conditions, can be translated into relational algebra with SQL null values expressions. 

\noindent
Conversely, with a similar argument any relational algebra with SQL null values expression can be translated in a 2VL translated \fosql query \citep{paredaens_structure_1989}.\qed
\end{proof}

\begin{example}
Let's consider the \textit{pseudo-identity} SQL query (1) from Section~\ref{sec:intro}:
{\footnotesize
\begin{lstlisting}
SELECT c1, c2 FROM r WHERE c1 = c1 AND c2 = c2;
\end{lstlisting}}
\noindent
Using the standard translation, its direct equivalent translation in \nRA is $(\select{1=1}{\select{2=2}{\texttt{r}}})$ as in Example~\ref{ex:nRA-identity}. As expected, both the SQL and the \nRA query give the same answer over any database instance.\qed	
\end{example}

The idea that the SQL null value can be encoded like any other constant as part of the data domain, and that equality and inequality fail for the special null value constant, as we stated in~\citep{franconi:tessaris:null:a:12,franconi:tessaris:null:b:12}, was known already in the database theory community, as it is witnessed in the works restricted to conjunctive queries by~\citet{nutt-et-al-null:07}, and in the works by \citeauthor{bravo:bertossi:06} (\citeyear{bravo:bertossi:06,BERTOSSI_2016}) in the context of peer-to-peer databases~\citep{DBLP:conf/edbtw/FranconiKLZ04}. Only recently there has been a complete formal account of it extended with bag semantics by~\citet{libkin_peterfreund_2020}. In this Section we have emphasised how a minimal extension of standard relational algebra can deal with SQL null values: the only change is in the behaviour of equality in the selection operator -- the absorbed negation normal form is considered in the mapping of derived operators as defined in Figures~\ref{fig:ra} and~\ref{fig:nra}. Moreover, the \nRA algebra, unlike the algebra introduced by~\citet{libkin_peterfreund_2020}, includes a \emph{singleton} operator and allows for \emph{zero-ary relations} and \emph{zero-ary projections}. The singleton operator covers the case of SQL queries introducing new domain constants in the SELECT clause. The singleton operator and the ability to have zero-ary relations and zero-ary projections are crucial to obtain the equivalence of the relational algebra with null values \nRA with the domain relational calculus with null values \FOLe that will be introduced in Section~\ref{sec:null_relational_calculus}.

\medskip
\citet{bultzingsloewen_translating_1987} and~\citet{negri_formal_1991} have proposed a direct translation of SQL with null values to a three-valued relational calculus (the former proposed a \emph{domain} relational calculus, the latter a \emph{tuple} relational calculus), with a model theory based on \InRA{\I} interpretations, where the null value is an element of the domain. These calculi mirror the relational algebra \nRA presented in this Section, by extending comparison operators to evaluate to \emph{unknown} whenever they involve a null value, introducing an isNull operator, and dealing accordingly with a three-valued semantics for Boolean operators and quantifiers. \citet{CONSOLE2022103603} show how to reduce the SQL three-valued relational calculus with the null value as an element of the domain, to a two-valued relational calculus, still with the null value as an element of the domain, extended with an \emph{assertion} operator necessary to evaluate the absorbed negation normal form of the Boolean combination of atomic conditions under classical 2VL semantics.

The aim of our research is to find a relational calculus, equivalent to \fosql, having a null value term in the language behaving as a marker for missing data, and therefore not appearing as an element of the domain of interpretation.

\section{Relational Calculus with Null Values} 
\label{sec:null_relational_calculus}\label{def:nFOL}

The \emph{Relational Calculus with Null Values} (\FOLe) is a first-order logic language with an explicit term $\nullvalue$ representing the null value, and with the interpretation of predicates as \textit{partial} tuples, namely predicates denote tuples over \emph{subsets} of the arguments instead of the whole set of arguments. Note that $\nullvalue$ is a term but not a constant, since the domain of interpretation in \FOLe does not include the null value: there is no object of the domain denoted by a null value, which indeed is meant to represent a missing value. \FOLe extends classical first-order logic in order to take into account the possibility that some of the arguments of a relation might not exist: this is accomplished by considering the interpretations of \FOLe to be \IFOLe{\I} interpretations.

Given a set of predicate symbols each one associated to an arity, and a set $\C$ of constants -- forming the relational schema (or \emph{signature}) $\R$ -- the equality binary predicate ``\foleqp'', a set of variable symbols, and the special \emph{null} symbol $\nullvalue$, terms of \FOLe are constants, variables, and $\nullvalue$, and formulae of \FOLe are defined by the following rules:
\begin{enumerate}
	\item if $t_1,\ldots,t_n$ are terms and $P$ is a predicate symbol in $\R$ of arity $n$,  $P(t_1,\ldots,t_n)$ is an atomic formula;
	\item if $t_1,t_2$ are terms different from $\nullvalue$,  $\foleq{t_1}{t_2}$ is an atomic formula;
	\item atomic formulae are formulae;
	\item if $\varphi$ and $\psi$ are formulae, then $\neg \varphi$, $\varphi \wedge \psi$ and $\varphi \vee \psi$ are formulae;
	\item if $\varphi$ is a formula and $x$ is a variable, then $\exists x. \varphi$ and $\forall x. \varphi$ are formulae.
\end{enumerate}

The semantics of \FOLe formulae is given in terms of database instances of type \IFOLe{\I}, called \emph{interpretations}. As such, differently from classical \FOL, the interpretation of a n-ary predicate symbol is a set of partial n-tuples: a partial n-tuple $\tau$ of values from a set $\Delta$ is a \textit{partial} function $\tau: \{1,\ldots,n\} \rightarrow\Delta$. In classical \FOL, the interpretation of a predicate symbol is instead a set of \textit{total} tuples.

An interpretation \IFOLe{\I} is composed by a non-empty domain $\Delta$ and an interpretation function ${\IFOLe{\I}}(\cdot)$, which interprets each constant symbol as an element of the domain, each predicate symbol as a set of partial tuples, and the equality predicate as the set of all pairs of equal elements of the domain $\Delta$. An assignment function gives meaning to variable symbols: an assignment $\alpha$ over a set of variable symbols $\cal V$ is a partial function $\alpha: {\cal V} \rightarrow\Delta$. If $\alpha$ is an assignment, $x\in{\cal V}$, and $o\in\Delta$, then $\replfunc{\alpha}{x}{o}$ is the assignment which maps $x$ to $o$ and agrees with $\alpha$ on all variables distinct from $x$.

Given an interpretation \IFOLe{\I}, an assignment $\alpha$, and a term $t$ (constant, variable, or $\nullvalue$), the interpretation of $t$ under \IFOLe{\I} and $\alpha$ - written $t^{\IFOLe{\I}\!,\,\alpha}$ - is defined as follows: 

\medskip
\begin{tabular}{l p{2em} l}
	$t^{\IFOLe{\I}\!,\,\alpha} = t^{\IFOLe{\I}}$   &&  if $t$ is a constant;\\
	$t^{\IFOLe{\I}\!,\,\alpha} = \alpha(t)$    &&   if $t$ is a variable in the domain of $\alpha$;\\
	$t^{\IFOLe{\I}\!,\,\alpha}$ undefined   && otherwise.
\end{tabular}
\smallskip

Given an interpretation \IFOLe{\I} and an n-ary predicate symbol $P\in\R$, the interpretation of $P$ - written $P^{\IFOLe{\I}}$ - is a set of partial n-tuples, and the interpretation of the binary equality predicate symbol ``$=$'' is the set of total 2-tuples $\{\{1\mapsto o,2\mapsto o\} | o\in\Delta\}$.

Given an interpretation \IFOLe{\I}, an assignment $\alpha$, and an n-ary atom $P(t_1, ... , t_n)$ with $t_1, ... , t_n$ an ordered sequence of $n$ terms (constants, variables, or $\nullvalue$), the interpretation \IFOLe{\I} satisfies the atom $P(t_1, ... , t_n)$ under the assignment $\alpha$ -- written $\IFOLe{\I}\!,\,\alpha\models P(t_1, ... , t_n)$ -- if and only if the partial n-tuple $\tau = \{ i\mapsto t_i^{\IFOLe{\I}\!,\,\alpha}\mid  \; t_i^{\IFOLe{\I}\!,\,\alpha} \;\text{defined}\;\text{for}\;i=1,\ldots,n\}$ is in the interpretation of $P$, namely $\tau\in P^{\IFOLe{\I}}$.

An interpretation $\IFOLe{\I}$ satisfies a non atomic formula $\varphi$ with an assignment $\alpha$ - written $\IFOLe{\I}\!,\,\alpha \models \varphi$ - if the following holds:
$$
\begin{array}{l@{\;\;\;\;}c@{\;\;\;\;}ll}
\IFOLe{\I}\!,\,\alpha \models \neg\varphi &\text{iff}&  \IFOLe{\I}\!,\,\alpha \not\models \varphi\\
\IFOLe{\I}\!,\,\alpha \models \varphi \land \psi &\text{iff}&  \IFOLe{\I}\!,\,\alpha \models \varphi  \text{\;\;and\;\;} \IFOLe{\I}\!,\,\alpha \models \psi\\
\IFOLe{\I}\!,\,\alpha \models \varphi \lor \psi &\text{iff} & \IFOLe{\I}\!,\,\alpha \models \varphi  \text{\;\;or\;\;}  \IFOLe{\I}\!,\,\alpha \models \psi\vspace{.5ex}\\
\IFOLe{\I}\!,\,\alpha \models \forall x.\varphi  &\text{iff} & \text{for all~} o\in\Delta: \IFOLe{\I},\replfunc{\alpha}{x}{o}\models\varphi\vspace{.5ex}\\
\IFOLe{\I}\!,\,\alpha \models \exists x.\varphi  &\text{iff} & \text{exists~} o\in\Delta: \IFOLe{\I},\replfunc{\alpha}{x}{o}\models\varphi
\end{array}
$$

As usual, an interpretation \IFOLe{\I} satisfying a closed formula is called a \emph{model} of the formula.
Satisfiability of a \FOLe closed formula without any occurrence of the null symbol $\nullvalue$ doesn't depend on partial tuples, and its models can be characterised by classical first-order semantics: in each model the interpretation of predicates would include only tuples represented as total functions. So, \FOLe without the null symbol $\nullvalue$ coincides with classical \FOL.

\begin{example}
The models of the \FOLe formula $\texttt{r}(\texttt{a},\texttt{a})\land \texttt{r}(\texttt{b},\nullvalue)$ are the interpretations $\IFOLe{\I}$ such that $\IFOLe{\I}(\texttt{r})$ includes the tuples $\{1\mapsto \texttt{a},2\mapsto \texttt{a}\}$ and $\{1\mapsto \texttt{b}\}$.\qed 
\end{example}


\section{The \emph{Decomposed} Relational Calculus with Null Values}
\label{sub:logical_reconstruction}

Given a database instance $\I$ with signature $\R$, let's consider a classical first-order logic language with equality (\FOL) over the \emph{decomposed} signature \subR, as it has been defined in Section~\ref{sec:preliminaries}. This language has a classical semantics with models of type \IRA{\I} and it does have a term to refer to null values. We call this logic the \textit{Decomposed Relational Calculus with Null Values}. In this Section we show that, given a database instance $\I$ with signature $\R$, \FOLe over $\R$ and \FOL over \subR{} are equally expressive: namely, for every formula in \FOLe over the signature $\R$ there is a corresponding formula in \FOL over the decomposed signature \subR, such that the two formulae have isomorphic models, and that for every formula in \FOL over the decomposed signature $\subR$ there is a corresponding formula in \FOLe over the signature $\R$, such that the two formulae have isomorphic models.
As we discussed before, the isomorphism between the interpretations \InRA{\I} and \IRA{\I} is based on the fact that each $|A|$-tuple in the relation $\IRA{\I}(\subR[R]_A)$ corresponds exactly to the $n$-tuple in $\InRA{\I}(R)$ having non-null values \emph{only} in the positional arguments specified in $A$, with the same values and in the same relative positional order.

This correspondence is interesting since it relates two relational calculi talking about database instances with null values $\I$, in neither of which the null value is an element of the domain of interpretation. The first calculus is a non-standard first order logic, since tuples correspond to partial functions; its signature is equal to the relational schema, and it reduces to standard first order logic in absence of null values. The latter calculus is a standard first order logic over a decomposition of the relational schema, and it also reduces to standard first order logic in absence of null values.

In order to relate the two logics, we define a bijective translation function $\FOLetoFOLd(\cdot)$ (and its inverse $\FOLdtoFOLe(\cdot)$) which maps \FOLe formulae into \FOL formulae (and vice versa). 

\begin{definition}[Bijective translation $\FOLetoFOLd$]
$\FOLetoFOLd(\cdot)$ is a bijective function from \FOLe formulae over the signature $\R$ to \FOL formulae over the signature $\subR$, defined as follows.
\begin{itemize}
\item 
$\FOLetoFOLd(R(t_1,\ldots,t_n)) =  \subR[R]_{\set{i_1,\ldots,i_k}}(t_{i_1},\ldots,t_{i_k})$,

~~where $R\in\R$ an $n$-ary relation, $\set{i_1,\ldots,i_k} = \set{j\in [1\cdots n]\mid t_j\text{ is not } \nullvalue}$.

\item $\FOLdtoFOLe(\subR[R]_{\set{i_1,\ldots,i_k}}(t_{i_1},\ldots,t_{i_k})) =  R(t'_1,\ldots,t'_n)$,

\hspace{1em}\begin{minipage}[t]{\textwidth}where $\set{i_1,\ldots,i_k}\subseteq[1\cdots n]$ and $t_{i_j} = t'_j$ for $j=1,\ldots,k$~,\\ and $t'_j$ is \nullvalue for $j\in[1\cdots n]\setminus \set{i_1,\ldots,i_k}$.\end{minipage}
\end{itemize}

\noindent
In both the direct and inverse cases we assume $i_1,\ldots,i_k$ in ascending order. 

\noindent
The translation of equality atoms and of non atomic formulae is the identity transformation inductively defined on top of the above translation of atomic formulae.\qed
\end{definition}

\begin{example}
The \FOLe formula $\exists x. R(\texttt{a},x)\land R(x,\nullvalue)\land R(\nullvalue,\nullvalue)$ over the signature $R$ is translated as the \FOL formula $\exists x. \subR[R]_{\{1,2\}}(\texttt{a},x)\land \subR[R]_{\{1\}}(x)\land \subR[R]_{\set{}}$ over the decomposed signature $\subR$, and vice versa.\qed
\end{example}

The above bijective translation preserves the models of the formulae, modulo the isomorphism among models presented in Section~\ref{sec:preliminaries}. We prove this in the following proposition.

\begin{proposition}\label{prop:equivalence_nRC_FOL}
	Let $\varphi$ be a \FOLe formula over the signature $\R$, and $\subR[\varphi]$ a \FOL formula over the signature $\subR$. Then for any (database) instance $\I$ and assignment $\alpha$ \emph{total w.r.t.\ free variables in $\varphi$}:
	$$\IFOLe{\I}\!,\,\alpha \models_{\FOLe} \varphi \text{~~~if and only if~~~} \IRA{\I}, \alpha \models_{\FOL} \FOLetoFOLd(\varphi).$$

	Since $\FOLetoFOLd(\cdot)$ is a bijection, we can equivalently write:
	$$\IFOLe{\I}\!,\,\alpha \models_{\FOLe} \FOLdtoFOLe(\subR[\varphi]) \text{~~~if and only if~~~} \IRA{\I}, \alpha \models_{\FOL} \subR[\varphi].$$
\end{proposition}
\begin{proof}

The statements of the theorem are proved by induction on the syntax of the formulae $\varphi$ and $\subR[\varphi]$. Since syntax and semantics of non-atomic formulae is the same for both \FOLe and \FOL, we show that the statement hold for the atomic cases.

The case of equality can be easily proved by noticing that $\nullvalue$ is not allowed to appear among its arguments and the fact that $\alpha$ is total w.r.t.\ the free variables that  might appear among the equality arguments.

We focus on ground atomic formulae, since the valuation function for variable symbols $\alpha$ is the same in the pre- and the post-conditions. 

\begin{enumerate}
\item $\left[\text{if }\IFOLe{\I}\!,\,\alpha \models_{\FOLe} \varphi\text{, then }\IRA{\I}, \alpha \models_{\FOL} \FOLetoFOLd(\varphi)\right]$

Let $R(t_1,\ldots,t_n)$ be a \FOLe ground atomic formula where $A = \set{i_1,\ldots,i_k}\subseteq [1\cdots n]$ is the set of its positional arguments for which $t_i$ is not $\nullvalue$. Then it holds that $\FOLetoFOLd(R(t_1,\ldots,t_n)) = \subR[R](t_{i_1},\ldots,t_{i_k})$. 

By the assumption $\IFOLe{\I}\!,\,\alpha \models_{\FOLe} R(t_1,\ldots,t_n)$; therefore there is a tuple $\tau\in {\IFOLe{\I}(R)}$ such that $\tau(j)=\IFOLe{\I}(t_j)$ for $j\in A$, and $\tau$ is undefined for $j\not\in A$. 
By definition of $\IRA{\I}$, $(\tau(i_1),\ldots,\tau(i_k))\in {\IRA{\I}(\subR[R])}$ therefore $\IRA{\I},\alpha \models \widetilde{R}_A(t_{i_1},\ldots,t_{i_k})$.

\item $\left[\text{if }\IRA{\I},\alpha \models_{\FOL} \subR[\varphi]\text{, then }\IFOLe{\I}, \alpha \models_{\FOLe} \FOLdtoFOLe(\subR[\varphi])\right]$

Let $\subR[R]_{\set{i_1,\ldots,i_k}}(t_1,\ldots,t_k)$ be a \FOLe ground atomic formula with $\set{i_1,\ldots,i_k}\subseteq [1\cdots n]$; then $\FOLdtoFOLe(\subR[R]_{\set{i_1,\ldots,i_k}}(t_1,\ldots,t_k)) = R(t'_1,\ldots,t'_n)$, where $\set{i_1,\ldots,i_k}\subseteq[1\cdots n]$ and $t'_{i_j} = t_j$ for $j=1,\ldots,k$ and $t'_j=\nullvalue$ for $j\in[1\cdots n]\setminus \set{i_1,\ldots,i_k}$. 

By assumption $\IRA{\I},\alpha \models \subR[R]_{\set{i_1,\ldots,i_k}}(t_1,\ldots,t_k)$, therefore $(\IRA{\I}(t_1),\ldots,\IRA{\I}(t_k)) \in\IRA{\I}(\subR[R]_{\set{i_1,\ldots,i_k}})$. By definition of $\IFOLe{\I}$ there is a tuple $\tau\in {\IFOLe{\I}}(R)$ s.t.\ $\tau(i_j) = \IRA{\I}(t_j)$ for $j\in \set{i_1,\ldots,i_k}$ and undefined otherwise. Therefore $\IFOLe{\I}\!,\,\alpha\models_{\FOLe} R(t'_1,\ldots,t'_n)$.
\end{enumerate}

For the inductive step we need to show that if the sub-formula has a free variable, then the assignment includes that variable in the domain. This is guaranteed by the fact that by hypothesis the initial assignment is total, so additional variables are bound by a quantifier in the formula; while the mapping $\FOLetoFOLd(\cdot)$ doesn't introduce additional variables. 
\qed
\end{proof}


\section{Queries in the Relational Calculi}
\label{sec:queries}

In this Section we introduce additional conditions over the logics we have defined so far, such that they can be used effectively with database instances. We first define the \textit{Standard Name Assumption}, to connect constants in the logic language to domain values in the database instances. We then define the notion of \textit{queries} over database instances, written as logical formulae. Finally, we introduce the \textit{domain independent} fragments of the logics, so to let query answers depend only on the actual content of the relations in a database instance. The definitions given in this Section are general, and they do apply to both $\FOLe$ and $\FOL$ over decomposed signature. 

The \textit{Standard Name Assumption} (SNA) in a logic states that, given a database instance $\I$, the interpretation of any constant $c\in\C$ is equal to itself, namely it holds that $c^{\I} = c$, therefore implying that $\C\subseteq\Delta$.
The \textit{Active Domain Inclusion} assumption states that, given a database instance $\I$, its active domain is included among the constants $\C$ of the logic, namely for each relation $R$ in the signature of the logic the image of each $\tau\in R^{\I}$ is included in $\C$: $$\bigcup_{R\in\R,\ \tau\in R^{\I}} \tau^{\rightarrow}\ \subseteq\C\ .$$

\noindent
 By using SNA together with the active domain inclusion assumption, we can interplay between constants in the logic language and values which are stored in database instances.

We can weaken the Standard Name Assumption by assuming \emph{Unique Names} instead.
An interpretation ${\I}$ satisfies the Unique Name Assumption (UNA) if $a^{{\I}} \neq b^{{\I}}$ for any different $a,b\in\C$. Note that SNA implies UNA, but not vice versa.
An interpretation is a model of a formula with the Standard Name Assumption \emph{if and only if} the interpretation obtained by homomorphically transforming the standard names with arbitrary domain elements is a model of the formula; this latter interpretation satisfies the Unique Name Assumption.
It is possible therefore to interchange the Standard Name and the Unique Name Assumptions; this is of practical advantage, since the Unique Name Assumption can be encoded in classical first-order logic with equality and it is natively present in most knowledge representation -- e.g., description logics -- theorem provers. 


A query in a logic with standard names and active domain inclusion is an open formula with a total order among the free variables appearing in the formula. The semantics of queries is defined by means of \emph{assignments} from the free variables to elements of the domain. The ordering is necessary to relate answers to tuples, and it may be defined by the the lexicographical order of the free variables, or by their order of appearance in the linear formula.

\begin{definition}[Answer tuple]
Given an assignment $\alpha$ over an ordered set of $n$ variable symbols $\{v_1,\ldots , v_n\}$, the answer tuple corresponding to $\alpha$ is the partial function $\anstup{\alpha}: \{1,\dots,n\} \rightarrow \Delta$ such that $\anstup{\alpha}(i) = \alpha (v_i)$ \qed
\end{definition}

Queries act as operators transforming a database instance into another database instance -- the latter composed by a single relation, corresponding to the answer to the query. 

\begin{definition}[Query answer]
Given a database instance ${\I}$, an answer to a query $\mathcal{Q}$ is the answer tuple $\anstup{\alpha}$ such that ${\I},\alpha \models\mathcal{Q}$.
The answer set $\eval[{\I}]{\mathcal{Q}}$ is the relation:
\begin{center}
$\eval[{\I}]{\mathcal{Q}}=\left\{\anstup{\alpha}\mid{\I},\alpha \models\mathcal{Q}\right\}$\qed
\end{center}
\end{definition}

\begin{example}
The \textit{pseudo-identity} \FOLe query 
$(\exists x_1,y_1.\ \texttt{r}(x_1,y_1)\land x=x_1\land y=y_1)$, with free variables $x$ and $y$,
evaluates to
$\{\{1\mapsto \texttt{a},2\mapsto \texttt{a}\}\}$ over the database instance $\IFOLe{\I}$ with $\IFOLe{\I}(\texttt{r})=\{\{1\mapsto \texttt{a},2\mapsto \texttt{a}\},\{1\mapsto \texttt{b}\}\}$, namely $\eval[\IFOLe{\I}]{\exists x_1,y_1.\ \texttt{r}(x_1,y_1)\land x=x_1\land y=y_1}=\{\{1\mapsto \texttt{a},2\mapsto \texttt{a}\}\}$. \qed	
\end{example}

\begin{definition}[Certain Answer] 
An answer tuple $\anstup{\alpha}$ is a certain answer of a query $\mathcal{Q}$ over a set $\mathcal{KB}$ of closed formulae if it is an answer to the query $\mathcal{Q}$ in each instance ${\I}$ satisfying $\mathcal{KB}$; i.e., ${\I},\alpha \models \mathcal{Q}$ for each ${\I}$ such that ${\I}\models \mathcal{KB}$. The certain answer set $\eval[\mathcal{KB}]{\mathcal{Q}}$ is the relation: 
\begin{center}
$\eval[\mathcal{KB}]{\mathcal{Q}}=\left\{\anstup{\alpha}\mid\forall\I.\ {\I}\models\mathcal{KB}\rightarrow{\I},\alpha \models\mathcal{Q}\right\}$\qed
\end{center}
\end{definition}

\begin{example}
The \textit{pseudo-identity} \FOLe query 
$(\exists x_1,y_1.\ \texttt{r}(x_1,y_1)\land x=x_1\land y=y_1)$, with free variables $x$ and $y$,
evaluates to
$\{\{1\mapsto \texttt{a},2\mapsto \texttt{a}\}\}$ over the knowledge base $\mathcal{KB}=\texttt{r}(\texttt{a},\texttt{a})\land\exists x.\ \texttt{r}(x,\nullvalue)$, namely $\eval[\mathcal{KB}]{\exists x_1,y_1.\ \texttt{r}(x_1,y_1)\land x=x_1\land y=y_1}=\{\{1\mapsto \texttt{a},2\mapsto \texttt{a}\}\}$. \qed	
\end{example} 


The following example shows that there are queries evaluating to unintuitive answers, since they evaluate to different answers over database instances having the same signature, the same interpretation to each relation symbol and constant symbol, but different domains.

\begin{example}
The query 
$\neg\texttt{p}(x)$, with free variable $x$,
evaluates to the empty set $\{\}$ over the database instance ${\I}$ with ${\I}(\texttt{p})=\{\{1\mapsto \texttt{a}\}\}$ and domain $\Delta^{\I}=\{\texttt{a}\}$, while it evaluates to $\{\{1\mapsto \texttt{b}\}\}$ over the database instance ${\J}$ with the same interpretation for \texttt{p} , namely ${\J}(\texttt{p})=\{\{1\mapsto \texttt{a}\}\}$, but different domain $\Delta^{J}=\{\texttt{a},\texttt{b}\}$. \qed
\end{example} 

To overcome this problem, we restrict queries to be only \textit{domain independent} formulae~\citep{DBLP:conf/cilc/KerhetF12}. 

\begin{definition}[Domain Independence]
A closed formula $\varphi$ is \textit{domain independent} if for every two interpretations
$\mathcal{I}=\langle\Delta^\mathcal{I},\mathcal{I}(\cdot)\rangle$ and 
$\mathcal{J}=\langle\Delta^\mathcal{J},\mathcal{J}(\cdot)\rangle$, which agree on the interpretation of relation symbols and constant symbols
-- i.e. $\mathcal{I}(\cdot)=\mathcal{J}(\cdot)$ -- but disagree on the interpretation domains $\Delta^\mathcal{I}$ and $\Delta^\mathcal{J}$:
$$\mathcal{I}\ \models\varphi \text{~~~if and only if~~~} 
\mathcal{J}\ \models\varphi.$$
The domain independent fragment of a logic includes only its domain independent formulae.\qed 
\end{definition}

It is well known that the domain independent fragment of the classical first-order logic \FOL can be characterised with its \emph{safe-range} syntactic fragment: intuitively, a formula is safe-range if and only if its variables are bounded by positive predicates or equalities -- for the exact syntactical definition see, e.g., \citep{AbHV95}.
It is easy to see that \textit{also} the domain independent fragment of \FOLe can be characterised with the \emph{safe-range} syntactic fragment of \FOLe. 
Due to the strong semantic equivalence expressed in Proposition~\ref{prop:equivalence_nRC_FOL} and to the fact the the bijection $\FOLetoFOLd(\cdot)$ preserves the syntactic structure of the formulae, we can reuse the results about safe-range transformations and domain independence holding for classical \FOL. 
To check whether a formula is safe-range, the formula is transformed into a logically equivalent \emph{safe-range normal form} and its \emph{range restriction} is computed according to a set of syntax based rules; the range restriction of a formula is a subset of its free variables, and if it coincides with the free variables then the formula is said to be safe-range \citep{AbHV95}. 
\begin{proposition}
Any safe-range \FOLe closed formula is domain independent, and any domain independent \FOLe closed formula can be transformed into a logically equivalent safe-range \FOLe closed formula. 
\end{proposition}

We observe that an interpretation is a model of a formula in the domain independent fragment of \FOLe with the Standard Name Assumption \emph{if and only if} the interpretation which agrees on the interpretation of relation and constant symbols but with the interpretation domain equal to the set of standard names $\C$ is a model of the formula. Therefore, in the following when dealing with the domain independent fragment of \FOLe with the Standard Name Assumption we can just consider interpretations with the interpretation domain equal to $\C$.

\section{Equivalence of Algebra and Calculus}
\label{sub:nra_and_nRC}

In this Section we prove the strong relation between \nRA and \FOLe by showing that the \nRA relational algebra with nulls and the domain independent fragment of \FOLe with the Standard Name Assumption are equally expressive.
The notion of \emph{equal expressivity} is captured by the following two propositions.

\begin{proposition}\label{prop:nRA_nFOL}
	Let $e$ be an arbitrary \nRA expression of arity $n$, and $t$ an arbitrary $n$-tuple as a total function with values taken from the set $\C \cup \{\nullvalue\}$. There is a function $\Omega(e,t)$ translating $e$ with respect to $t$ into a closed safe-range \FOLe formula, such that for any instance $\I$ with the Standard Name Assumption:
	\begin{center}
	~~~$t\in \eval[\InRA{\I}]{e}  \text{~~~if and only if~~~} \IFOLe{\I} \models_{\FOLe} \Omega(e,t) .$
	\end{center}
\end{proposition}

\begin{proposition}\label{prop:nFOL_to_nRA}
	Let $\varphi$ be an arbitrary safe-range \FOLe closed formula. There is a \nRA expression $e$, such that for any instance $\I$ with the Standard Name Assumption:
	\begin{center}
	~~~$\IFOLe{\I} \models_{\FOLe} \varphi \text{~~~if and only if~~~}  \eval[\InRA{\I}]{e} \neq \emptyset .$
	\end{center}
\end{proposition}

\noindent
These lemmata enable to state the main theorem as:
\begin{theorem}[Extension of Codd's Theorem]
	The \nRA relational algebra with nulls and the domain independent fragment of \FOLe with the Standard Name Assumption are equally expressive.
\end{theorem}

This means that there exists a reduction from the membership problem of a tuple in the answer of a \nRA expression over a database instance with null values into the satisfiability problem of a closed safe-range \FOLe formula over the same database (modulo the isomorphism among database instances presented in Section~\ref{sec:preliminaries}); and there exists a reduction from the satisfiability problem of a closed safe-range \FOLe formula over a database instance with null values into the emptiness problem of the answer of a \nRA expression over the same database (modulo the isomorphism among database instances).


\subsection{From Algebra to Calculus}
\label{sec:alg-calc}
	
To show that the safe-range fragment of \FOLe with the Standard Name Assumption captures the expressivity of \nRA queries, we first define explicitly a translation function which maps \nRA expressions into safe-range \FOLe formulae.

\begin{definition}[From \nRA to safe-range \FOLe]\label{def:nRAtonFOL_func}
	Let $e$ be an arbitrary \nRA expression, and $t$ an arbitrary tuple of the same arity as $e$, as a total function with values taken from the set $\C \cup \{\nullvalue\} \cup \V$, where $\V$ is a countable set of variable symbols. The function $\Omega(e,t)$ translates $e$ with respect to $t$ into a \FOLe formula according to the following inductive definition:
	
	\begin{itemize}
		\item for any $R/_{\ell} \in \R$, $\Omega(R,t) \leadsto R(t(1),\ldots,t(\ell))$
		
		\item $\Omega(\singleton{v},t) \leadsto 
		\begin{cases}
		\foleq{t(1)}{v} & \text{if $t(1)\neq\nullvalue$}\\
		\false & \text{otherwise}		
		\end{cases}$
		
		\item $\Omega(\singleton{\nullvalue},t) \leadsto 
		\begin{cases}
		\false & \text{if $t(1)\neq\nullvalue$}\\
		\true & \text{otherwise}\\
		\end{cases}$

		\item $\Omega(\select{\textsf{i = v}}{e}, t) \leadsto 		 		\begin{cases}
		\Omega({e}, t) \wedge \foleq{t(i)}{v}  & \text{if $t(i)\neq\nullvalue$}\\
		\false & \text{otherwise}	
		\end{cases}$

		\item $\Omega(\select{\textsf{i = j}}{e}, t) \leadsto 
		\begin{cases}
		\Omega({e}, t) \wedge \foleq{t(i)}{t(j)} & \text{if $t(i),t(j)\neq\nullvalue$}\\
		\false & \text{otherwise}	
		\end{cases}$

		\item $\Omega(\project{i_1\cdots i_k}{e},t) \leadsto 
		\exists x_1\cdots x_n \bigvee_{H\subseteq {\set{1\cdots n}\setminus\set{i_1\cdots i_k}}} \Omega(e,t_H)$
		
		\vspace{1ex}
		where $x_i$ are fresh variable symbols and $t_H$ is a sequence of $n$ terms defined as:
		$$t_H(i) \doteq 
		\begin{cases}
		t(i) & \text{if $i\in\set{i_1,\ldots,i_k}$}\\
		\nullvalue & \text{if $i\in H$}\\
		x_i & \text{otherwise}
		\end{cases}$$
		
		\item $\Omega(e_1 \times e_2, t) \leadsto 
		\Omega(e_1,t') \wedge \Omega(e_2,t'')$ \\ 
		where $n_1, n_2$ are the arity of $e_1, e_2$ respectively, \\ 
				where $t'$ is a $n_1$-ary tuple function s.t. $t'(i)=t(i)$ for $1\leq i \leq n_1$, \\
				and $t''$ is a $n_2$-ary tuple function s.t.\ $t''(i) = t(n_1+i)$ for $1\leq i\leq n_2$						
		\vspace{1ex}
		\item $\Omega(e_1 \cup e_2, t) \leadsto 
		\Omega(e_1,t) \vee \Omega(e_2,t)$

		\item $\Omega(e_1 - e_2, t) \leadsto 
		\Omega(e_1,t) \wedge \neg \Omega(e_2,t)$\qed
	\end{itemize}
\end{definition}

We can now show the proof of Proposition~\ref{prop:nRA_nFOL}.

\begin{proof}[Proposition~\ref{prop:nRA_nFOL}]
To prove the proposition we show that for any expression $e$, tuple $t$ of the same arity, and instance $\I$ with the Standard Name Assumption the above defined function $\Omega$ satisfies the property that
$$t\in \eval[\InRA{\I}]{e}\text{ iff }\IFOLe{\I} \models_{\FOLe} \Omega(e,t).$$
We shall prove the property by induction on the structure of $e$, were the base cases are expressions of the form $R, \singleton{v}, \singleton{\nullvalue}$.

\begin{description}
\item[$R$] \hfill \\
  Let $e$ be $R\in\R$ of arity $n$ and $t$ a tuple s.t.\ $t\in \eval[\InRA{\I}]{R}$ where $\set{i_1,\ldots,i_k}$ are the indexes for which $t(i)=\nullvalue$. By Definition~\ref{def:nRAtonFOL_func}, $\Omega(R,t)\leadsto R(t)$.

  Since $t\in \eval[\InRA{\I}]{R}$, then $t\in \InRA{\I}(R)$, so there is a partial function $t' \in {\IFOLe{\I}}(R)$ s.t.\ $t'(i)$ is undefined for $i\in \set{i_1,\ldots,i_k}$ and $t'(i) = t(i)$ otherwise. Therefore $\IFOLe{\I}\models_{\FOLe}R(t)$.

  On the other hand, if $\IFOLe{\I}\models_{\FOLe}R(t)$ then there is a partial function $t' \in {\IFOLe{\I}}(R)$ s.t.\ $t'(i)$ is undefined for $i\in \set{i_1,\ldots,i_k}$ and $t'(i) = \IFOLe{\I}(t(i))$ otherwise. Since $\I$ is an interpretation with the Standard Name Assumption, ${\IFOLe{\I}}(t(i)) = t(i)$; therefore $t\in\InRA{\I}(R)$ and $t\in \eval[\InRA{\I}]{R}$.

\item[$\singleton{v}$]  \hfill \\
  $t\in \eval[\InRA{\I}]{\singleton{v}}$ iff $t(1) \neq \nullvalue$ and $t(1) = v$ iff $\models_{\FOLe} \foleq{t(1)}{v}$ and $t(1) \neq \nullvalue$, which is exactly the definition of $\Omega(\singleton{v},t)$ (note that it doesn't depend on the actual instance $\I$).

\item[$\singleton{\nullvalue}$]  \hfill \\
  this case is analogous to the previous one: $t\in \eval[\InRA{\I}]{\singleton{\nullvalue}}$ iff $t(1) = \nullvalue$ iff $\models_{\FOLe} \Omega(\singleton{\nullvalue},t)$.

\end{description}

For the inductive step we assume that the property holds for expressions $e, e_1, e_2$ and show that holds for the following composed expressions as well.

\begin{description}
\item[$\select{\textsf{i = v}}{e}$]  \hfill \\
  We consider two cases: $t(i)=\nullvalue$ and $t(i)\neq\nullvalue$.
  
  If $t(i)=\nullvalue$ then $t \not\in \eval[\InRA{\I}]{\select{\textsf{i = v}}{e}}$, which is equivalent to the translation $\Omega(\select{\textsf{i = v}}{e},t)\leadsto \false$.
  
  If $t(i)\neq\nullvalue$, the translation is $\Omega({e}, t) \wedge \foleq{t(i)}{v}$ and $\IFOLe{\I} \models_{\FOLe} \Omega(\select{\textsf{i = v}}{e},t)$ iff $\IFOLe{\I} \models_{\FOLe} \Omega({e}, t)$ and $\IFOLe{\I} \models_{\FOLe} \foleq{t(i)}{v}$ because there are no shared variables. By inductive hypothesis $\IFOLe{\I} \models_{\FOLe} \Omega({e}, t)$ iff $t \in \eval[\InRA{\I}]{e}$ and -- by definition of equality -- $\IFOLe{\I} \models_{\FOLe} \foleq{t(i)}{v}$ iff $t(i) = v$. $t \in \eval[\InRA{\I}]{e}$ and $t(i) = v$ is the necessary and sufficient condition for $t \in\eval[\InRA{\I}]{\select{\textsf{i = v}}{e}}$, therefore is equivalent to $\IFOLe{\I} \models_{\FOLe} \Omega({\select{\textsf{i = v}}{e}}, t)$.

\item[$\select{\textsf{i = j}}{e}$]  \hfill \\
  We consider the two cases: either $t(i)$ or $t(j)$ is equal to $\nullvalue$ or they are both different from $t(i)=\nullvalue$.

  In the first case for any $e$ and $\InRA{\I}$ $t\not\in \eval[\InRA{\I}]{\select{\textsf{i = j}}{e}}$ because of the constraint that both $i,j$ elements of the tuple must be different from $\nullvalue$. Therefore -- given that the translation is $\false$ -- $\IFOLe{\I} \models_{\FOLe} \false$ iff $t\in \eval[\InRA{\I}]{\select{\textsf{i = j}}{e}}$ is satisfied.
  
  In the second case, let's consider the translation $\Omega({e}, t) \wedge \foleq{t(i)}{t(j)}$; since there are no shared variables between the two expressions ($t$ contains constants only), then $\IFOLe{\I}\models_{\FOLe}\Omega({e}, t) \wedge \foleq{t(i)}{t(j)}$ iff $\IFOLe{\I}\models_{\FOLe}\Omega({e}, t)$ and $\IFOLe{\I}\models_{\FOLe}\foleq{t(i)}{t(j)}$. By inductive hypothesis $\IFOLe{\I}\models_{\FOLe}\foleq{t(i)}{t(j)}$ iff $t\in\eval[\InRA{\I}]{e}$ and, by definition of equality, $\models_{\FOLe}\foleq{t(i)}{t(j)}$ iff $t(i) = t(j)$; therefore $\IFOLe{\I}\models_{\FOLe}\Omega({e}, t) \wedge \foleq{t(i)}{t(j)}$ iff $t\in\eval[\InRA{\I}]{\select{\textsf{i = j}}{e}}$.

\item[$\project{i_1\cdots i_k}{e}$]  \hfill \\
Let $t\in\eval[\InRA{\I}]{\project{i_1\cdots i_k}{e}}$, then there is $t'$ of arity $n$ s.t.\ $t'\in \eval[\InRA{\I}]{e}$ and $t(j) = t'(i_j)$ for $i=1\ldots k$. By the recursive hypothesis, $\IFOLe{\I}\models_{\FOLe}\Omega(e,t')$. Let $H' \subseteq [1\cdots n]\setminus\set{i_1,\cdots, i_k}$ be the sets of indexes for which $t'(\cdot)$ is $\nullvalue$; then $\IFOLe{\I}\models_{\FOLe}\exists x_1\cdots x_n \Omega(e,t_{H'})$, so $\IFOLe{\I}\models_{\FOLe}\exists x_1\cdots x_n \bigvee_{H\subseteq {\set{1\cdots n}\setminus\set{i_1\cdots i_k}}} \Omega(e,t_H)$ because $H'$ is one of these disjuncts.

For the other direction, let us assume that $\IFOLe{\I}\models_{\FOLe}\Omega(\project{i_1\cdots i_k}{e},t)$, then $\IFOLe{\I}\models_{\FOLe}\exists x_1\cdots x_n \bigvee_{H\subseteq {[1\cdots n]\setminus\set{i_1\cdots i_k}}} \Omega(e,t_H)$, so there is $H'\subseteq {[1\cdots n]\setminus\set{i_1\cdots i_k}}$ and assignment $\alpha$ for variables in $t_{H'}$ s.t.\ $\IFOLe{\I}\models_{\FOLe} \Omega(e,\alpha(t_{H'}))$. By the inductive hypothesis there is a tuple $t'$ corresponding to $\alpha(t_{H'})$ s.t.\ $t'\in \eval[\InRA{\I}]{e}$; by construction of $t_{H'}$, $t(i) = \alpha(t_{H'}(i))$ for $i\in\set{i_1,\ldots, i_k}$ therefore $t\in\eval[\InRA{\I}]{\project{i_1\cdots i_k}{e}}$.

\item[$e_1 \times e_2$]  \hfill \\
  $\IFOLe{\I} \models_{\FOLe} \Omega({e_1 \times e_2},t)$ iff there are $t_1,t_2$ having the same arity of $e_1, e_2$ respectively (let them be $n_1$ and $n_2$) with $t_1(i) = t(i)$ and $t_2(j)=t(n_1+j)$, such that $\IFOLe{\I} \models_{\FOLe} \Omega({e_1},t_1)\land \Omega({e_2},t_2)$. Since there are no shared variables, the latter is equivalent to $\IFOLe{\I} \models_{\FOLe} \Omega({e_1},t_1)$ and $\IFOLe{\I} \models_{\FOLe} \Omega({e_2},t_2)$; by inductive hypothesis those are true iff $t_1\in\eval[\InRA{\I}]{e_1}$ and $t_2\in\eval[\InRA{\I}]{e_2}$. By definition, the latter is true iff $t\in \eval[\InRA{\I}]{e_1 \times e_2}$.

\item[$e_1 \cup e_2$]  \hfill \\
  $\IFOLe{\I} \models_{\FOLe} \Omega({e_1 \cup e_2},t)$ iff $\IFOLe{\I} \models_{\FOLe} \Omega({e_1},t)\lor \Omega({e_2},t)$ which, since there are no shared variables, is equivalent to $\IFOLe{\I} \models_{\FOLe} \Omega({e_1},t)$ or $\IFOLe{\I} \models_{\FOLe} \Omega({e_2},t)$. By inductive hypothesis the latter is true iff $t\in\eval[\InRA{\I}]{e_1}$ or $t\in\eval[\InRA{\I}]{e_2}$, which -- by definition -- is true iff $t\in\eval[\InRA{\I}]{e_1 \cup e_2}$.

\item[$e_1 - e_2$]  \hfill \\
 $\IFOLe{\I} \models_{\FOLe} \Omega({e_1 - e_2},t)$ iff $\IFOLe{\I} \models_{\FOLe} \Omega({e_1},t)\land \neg \Omega({e_2},t)$ which, since there are no shared variables, is equivalent to $\IFOLe{\I} \models_{\FOLe} \Omega({e_1},t)$ and $\IFOLe{\I} \not\models_{\FOLe} \Omega({e_2},t)$. By inductive hypothesis the latter is true iff $t\in\eval[\InRA{\I}]{e_1}$ and $t\not\in\eval[\InRA{\I}]{e_2}$, which -- by definition -- is true iff $t\in\eval[\InRA{\I}]{e_1 - e_2}$.\qed
\end{description}
\end{proof}

\begin{example}
Given some database instance with a binary relation \texttt{r}, checking whether the tuple $\langle \texttt{a},\texttt{a}\rangle$ or the tuple $\langle \texttt{b},\nullvalue\rangle$ are in the answer of the \nRA pseudo-identity query $\select{1=1}{\select{2=2}{\texttt{r}}}$  (discussed in Section~\ref{sec:intro}) corresponds to check whether the \FOLe closed safe-range formula $\Omega(\select{1=1}{\select{2=2}{\texttt{r}}},t)$, for $t=\langle \texttt{a},\texttt{a}\rangle$ or $t=\langle \texttt{b},\nullvalue\rangle$, is satisfied by the database instance. \\
In the case of $t=\langle \texttt{a},\texttt{a}\rangle$, using the Definition~\ref{def:nRAtonFOL_func} we have that the formula is equivalent to $\Omega({\select{2=2}{\texttt{r}}}, \langle \texttt{a},\texttt{a}\rangle) \land \foleq{a}{a}$, which is equivalent to $\Omega({\texttt{r}}, \langle \texttt{a},\texttt{a}\rangle) \land \foleq{a}{a}$, and finally to $\texttt{r}(\texttt{a},\texttt{a})$, which is satisfied by a database instance only if the table \texttt{r} contains the tuple $\langle \texttt{a},\texttt{a}\rangle$. \\
In the $\langle \texttt{b},\nullvalue\rangle$ case, the formula is equivalent to $\Omega({\select{2=2}{\texttt{r}}}, \langle \texttt{b},\nullvalue\rangle) \land \foleq{b}{b}$, and the formula $\Omega({\select{2=2}{\texttt{r}}}, \langle \texttt{b},\nullvalue\rangle)$ is equivalent to $\false$ by Definition~\ref{def:nRAtonFOL_func}; so, as expected, also when translated in first-order logic, it turns out that the tuple $\langle \texttt{b},\nullvalue\rangle$ (like any other tuple with at least a null value in it) can not be in the answer of the query $\select{1=1}{\select{2=2}{\texttt{r}}}$ for any database instance. %
\qed
\end{example}

\begin{example}
Let's consider the NOT-NULL constraint from Example~\ref{ex:nRA-denial}, expressed in \nRA as the denial constraint $(\select{\isnull(1)}{\texttt{q}}=\emptyset)$, in the database instance $\IFOLe{\I}$ where $\texttt{q}^{\IFOLe{\I}}=\{\langle a,a \rangle, \langle b,\nullvalue \rangle\}$. \\
In order to prove that the denial constraint is satisfied in $\IFOLe{\I}$, we have to prove that, for any tuple $t$, the safe range \FOLe formula $\Omega(\select{\isnull(1)}{\texttt{q}},t)$ is false in the database instance $\IFOLe{\I}$. 
The $\select{\isnull(i)}{e}$ operator is derived, and it is equivalent to $e - {\select{i = i}{e}}$ (see Figure~\ref{def:nRA}); therefore we need to consider the formula $\Omega(\texttt{q} - {\select{1=1}{\texttt{q}}},t)$.
Using Definition~\ref{def:nRAtonFOL_func}, we have:
$$
\Omega(\select{\isnull(1)}{\texttt{q}},t)\equiv
\Omega(\texttt{q} - \select{1=1}{\texttt{q}},t)\equiv
\Omega(\texttt{q},t)\land\neg\Omega(\select{1=1}{\texttt{q}},t)
$$
The formula $\Omega(\select{1=1}{\texttt{q}},t)$ depends on $t$; therefore, we have to consider all the possible cases for $t$. This amounts to check that none of the following formulas is satisfied by the database instance $\IFOLe{\I}$:
\noindent
$$
\IFOLe{\I}\nvDash
\begin{cases}
	\Omega(\texttt{q},t)\land\neg(\Omega(\texttt{q},t)\land\true)\equiv\false & {\text{if } t\in\texttt{q}^{\IFOLe{\I}}\text{ and hence  }\ t(1)\neq\nullvalue}\\ 
	\Omega(\texttt{q},t)\land\neg(\Omega(\texttt{q},t)\land\true)\equiv\false & {\text{if } t\notin\texttt{q}^{\IFOLe{\I}}\!\land\ t(1)\neq\nullvalue} \\
	\Omega(\texttt{q},t)\land\neg\false\equiv\texttt{q}(\nullvalue,t(2)) & {\text{if }t(1)=\nullvalue\ \text{  and hence  }\ t\notin\texttt{q}^{\IFOLe{\I}}}
\end{cases} 
$$
\noindent
It is easy to see that $\IFOLe{\I}\nvDash\Omega(\select{\isnull(1)}{\texttt{q}},t)$ for any $t$, namely for each of the three cases above.\qed
\end{example}


\subsection{From Calculus to Algebra}
\label{sec:calc-alg}

To show that \nRA queries capture the expressivity of the safe-range fragment of \FOLe with the Standard Name Assumption, we first establish two intermediate results supporting that \FOLe queries can be expressed in (standard) relational algebra.

\begin{lemma} \label{lemma:RAtonRA}
	Let $\subR[e]$ be an expression in the standard relational algebra \RA over the decomposed relational schema $\subR$. There is a \nRA expression $e$ over the relational schema $\R$, such that for any instance $\I$:
	\begin{displaymath}
	\eval[\IRA{\I}]{\subR[e]} = \eval[\InRA{\I}]{e}
	\end{displaymath}
\end{lemma}

\begin{proof}
	The \nRA expression $e$ is the same expression $\subR[e]$ where each basic relation $\subR[R]_{\set{i_1,\ldots,i_k}}$ with $R$ of arity $n$ is substituted by the expression
	\begin{displaymath}
	\subR[R]_{\set{i_1,\ldots,i_k}} ~{\leadsto}~ \project{i_1,\ldots,i_k}(\select{\isnotnull(\set{i_1,\ldots,i_k})}(\select{\isnull([1\cdots n]\setminus\set{i_1,\ldots,i_k})}R))
	\end{displaymath}
	where $\select{\isnull(i_1,\ldots,i_k)}e$ is a shorthand for $\select{\isnull(i_1)}\ldots\select{\isnull(i_k)}e$.
	\\
	The proof is by induction on the expression $e$ where the basic cases are the unary singleton - which is the same as \nRA since there are no nulls - and any basic relation $\subR[R]_{\set{i_1,\ldots,i_k}}$ where $R\in \R$ an $n$-ary relation. Compound expressions satisfy the equality because of the inductive assumption and the fact that the algebraic operators (in absence of null values) in \nRA and \RA have the same definition.
	\\
	Therefore we just need to show that
	\begin{displaymath}
	\eval[\IRA{\I}]{\subR[R]_{\set{i_1,\ldots,i_k}}} = \eval[\InRA{\I}]{\project{i_1,\ldots,i_k}(\select{\isnotnull(i_1,\ldots,i_k)}(\select{\isnull(j_1,\ldots, j_\ell)}R))}
	\end{displaymath}
	however, is true by the definition of $\IRA{\I}$ and $\InRA{\I}$ as shown in Section~\ref{sec:preliminaries}:
	
	$\IRA{\I}(\subR[R]_{\set{i_1,\ldots,i_k}})=
	\{t\in
	\begin{aligned}[t]
	&
	([1\cdots k]\mapsto\C)\mid \\
	&
	\exists t'\in\InRA{\I}(R).\ 
	\begin{aligned}[t]
	&\forall j\in([1\cdots n]\setminus\set{i_1,\ldots,i_k}).\ t'(j)=\nullvalue\land{}\\
	&\forall j\in\set{i_1,\ldots,i_k}.\ t'(j)\neq\nullvalue\land{}\\
	&\forall j\in\set{1,\ldots,k}.\ t(j)=t'(i_j)
	\}. 
    \end{aligned}
	\end{aligned} 
	$
	\vspace{-1.5ex}
	\\
	\qed
\end{proof}

\begin{lemma}\label{lemma:FOLtoRA}
	Let $\psi$ be a safe-range \FOL closed formula. There is an \RA expression $\subR[e]$ of arity zero, such that for any instance $\IRA{\I}$ with the Standard Name Assumption:
	\begin{displaymath}
	\IRA{\I} \models_{\FOL} \psi  \text{~~~if and only if~~~}  \eval[\IRA{\I}]{\subR[e]} \neq \emptyset
	\end{displaymath}
\end{lemma}

\begin{proof}
Since null values are not present neither in \FOL nor \RA, the lemma derives from the equivalence between safe-range relational calculus queries and relational algebra (see e.g.~\citep{AbHV95}). \qed
\end{proof}

Now we can prove Proposition~\ref{prop:nFOL_to_nRA} using the above lemmata.

\begin{proof}[Proposition~\ref{prop:nFOL_to_nRA}]
Given a safe-range \FOLe closed formula $\varphi$, by Proposition~\ref{prop:equivalence_nRC_FOL}, there is a safe-range \FOL formula $\subR[\varphi]$ such that $\IFOLe{\I} \models_{\FOLe} \varphi$ iff $\IRA{\I} \models_{\FOL}\subR[\varphi]$ -- since $\varphi$ is closed, then any assignment is total.
\\
By restricting to instances with the Standard Name Assumption, we can use Lemma~\ref{lemma:FOLtoRA} to conclude that there is an \RA expression $\subR[e]$  of arity zero such that $\IRA{\I} \models_{\FOL}\subR[\varphi]$ iff $\eval[\IRA{\I}]{\subR[e]} \neq \emptyset$; therefore $\IFOLe{\I} \models_{\FOLe} \varphi$ iff $\eval[\IRA{\I}]{\subR[e]} \neq \emptyset$
\\
Finally we use Lemma~\ref{lemma:RAtonRA} to show that there is an \nRA expression $e$ such that $\eval[\IRA{\I}]{\subR[e]} \neq \emptyset$ iff $\eval[\InRA{\I}]{e} \neq \emptyset$; which enables us to conclude that $\IFOLe{\I} \models_{\FOLe} \varphi$ iff $\eval[\InRA{\I}]{e} \neq \emptyset$. \qed
\end{proof}


\begin{example}
The safe-range closed \FOLe formula $\exists x. \texttt{r}(x,\nullvalue)$ is translated as the zero-ary \nRA statement $\select{\isnotnull(1)}{\select{\isnull(2)}{\texttt{r}}} \neq \emptyset$.\qed
\end{example}

\begin{example}
Let's consider what would be the classical way to express the FOREIGN KEY constraint from $\texttt{p.2}$ to $\texttt{q.1}$ from Example~\ref{ex:nRA-denial} in first-order logic: 

$\forall x,y.\ \texttt{p}(x,y)\rightarrow\exists z.\ \texttt{q}(y,z)
~~~\equiv~~~
\neg\exists y.\ (\exists x.\ \texttt{p}(x,y))\land\neg(\exists z.\ \texttt{q}(y,z))$. 

\noindent
The formula is translated as the \nRA statement: $\project{2}{\select{\isnotnull(2)}{\texttt{p}}} - \project{1}{\select{\isnotnull(1)}{\texttt{q}}}\ =\emptyset$, which is the \nRA statement we considered in Example~\ref{ex:nRA-denial}. \qed
\end{example}


\section{Semantic Integrity with Null Values in SQL:1999}

In this Section we consider the main integrity constraints involving a specific behaviour for null values as defined in the ANSI/ISO/IEC standard SQL:1999~\citep{melton_sql1999_2002}, and we show how these can be naturally captured using \FOLe. We focus on \emph{unique and primary key constraints} and on \emph{foreign key constraints} as defined in \citep{TuGe:VLDB:01a}. We will see how the actual definitions of the unique, primary key, and foreign key constraints are a bit more involved than we have seen before: this is because more complex cases than the simple examples above may happen involving null values.

\subsection{Unique and primary key constraints.}

As specified in \citep{TuGe:VLDB:01a}, a uniqueness constraint $\textbf{\textsf{UNIQUE}}(u_1,\ldots,u_n)$ holds for a table $R$ of arity $m>n$ in a database if and only if there are no two rows $r_1,r_2$ in $R$ such that the values of all their uniqueness columns $u_i$ match and are not null. More formally, in \emph{tuple} relational calculus with explicit null values in the domain (fixing the incomplete definition in \citep{TuGe:VLDB:01a}):
$$\forall r_1,r_2\in R.\ \left(r_1\neq r_2\land\bigwedge_{i=1}^n r_1.u_i\neq\nullvalue\land r_2.u_i\neq\nullvalue\right)\rightarrow\left(\bigvee_{i=1}^n r_1.u_i\neq r_2.u_i\right).$$
\noindent
In \FOLe, which is a \emph{domain} relational calculus, this constraint can be written as:

\medskip
\noindent
{\small
$\forall \ol{x}.\ \left(
    \bigwedge\limits_{\set{i_1\cdots i_k} \subseteq\set{1\cdots m}\setminus\set{u_1\cdots u_n}}
	\forall \ol{y}, \ol{z}.\ (\nullify{u_1,\ldots, u_n,i_1,\ldots,i_k}{\ol{x},\ol{y}}{R}\land\nullify{u_1,\ldots, u_n,i_1,\ldots,i_k}{\ol{x},\ol{z}}{R}) \rightarrow \ol{y}=\ol{z}
\right)\land\ $
	
\medskip
\noindent
	$
	\forall \ol{x}.\ \left(\bigwedge\limits_{\set{i_1\cdots i_k}\subseteq \set{1\cdots m}\setminus\set{u_1\cdots u_n}} \exists \ol{y}.\ \nullify{u_1,\ldots, u_n,i_1,\ldots,i_k}{\ol{x},\ol{y}}{R}\right) \rightarrow\bot
	$
	}

\medskip\noindent
where $\nullify{i_1,\ldots,i_k}{t_1,\ldots,t_k}{R}$ is a shorthand for $R(t'_1,\ldots,t'_m)$ where $t'_{j} = t_j$ for $j=i_1\ldots i_k$ and $t'_j = \nullvalue$ for all other $j$.

A \emph{primary key} constraint is a combination of a \emph{uniqueness} constraint and one or more \emph{not null} constraints. A constraint $\textbf{\textsf{PRIMARY KEY}}(u_1,\ldots,u_n)$ holds for a table $R$ if and only if the following holds, in tuple-relational calculus with explicit null values in the domain (fixing the incomplete definition in \citep{TuGe:VLDB:01a}):

$$\forall r\in R.\ \left(\bigwedge_{i=1}^n r.u_i\neq\nullvalue\right)\land\forall r_1,r_2\in R.\ r_1\neq r_2\rightarrow\left(\bigvee_{i=1}^n r_1.u_i\neq r_2.u_i\right).$$

\noindent
Note that no column shall occur more than once within the same unique/primary key definition; furthermore, each table can have at most one primary key.

\noindent
In \FOLe this constraint can be written as the conjunction of a uniqueness constraint as above and a not null constraint as follows for each key attribute $u_i$:

$$\neg\left(
\bigvee_{\set{i_1\cdots i_k}\subseteq\set{1\cdots m}\setminus\set{u_i}} \exists \ol{y}.\ \nullify{i_1,\ldots,i_k}{\ol{y}}{R}\right).
$$

\subsection{Foreign key constraints.}

Foreign key constraints (or referential constraints) express dependencies among (portions of) rows in tables.
Given a \emph{referenced} (or \emph{parent}) table and a \emph{referencing} (or \emph{child}) table, a subset $f_i,\dots,f_k$ of the columns of the referencing table builds the foreign key and refers to the unique/primary key columns $u_j,\ldots,u_l$ of the referenced table. As specified in \citep{TuGe:VLDB:01a}, the \emph{simple match} is the default foreign key constraint implemented by all DBMS vendors. For each row $r$ of the referencing table $R$ (child table), either at least one of the values of the referencing columns $f_1,\ldots,f_n$ is a null value or the value of each referencing columns $f_i, 1\leq i\leq n$, is equal to the value of the corresponding referenced column $u_i$ for some row $s$ of the referenced table $S$. More formally, in tuple-relational calculus with null values in the domain \citep{TuGe:VLDB:01a}:

$$\forall r\in R.\ \left(\bigwedge_{i=1}^n r.f_i\neq\nullvalue\right)\rightarrow\exists s\in S.\ \left(\bigwedge_{i=1}^n r.f_i=s.u_i\right).$$

\noindent
In \FOLe this constraint can be written as:

\medskip\noindent
\begin{center}
\begin{minipage}{.7\textwidth}
$\forall \ol{x}.\ 
\left(
\bigvee\limits_{\set{i_1\cdots i_k}\subseteq\set{1\cdots m}\setminus\set{u_1\cdots u_n}} \exists \ol{y}.\ \nullify{u_1,\ldots, u_n,i_1,\ldots,i_k}{\ol{x},\ol{y}}{S}\right)
\rightarrow$

\noindent
$\hspace{2em}
\left(
\bigvee\limits_{\set{i_1\cdots i_k}\subseteq\set{1\cdots m}\setminus\set{u_1\cdots u_n}} \exists \ol{y}.\ \nullify{u_1,\ldots, u_n,i_1,\ldots,i_k}{\ol{x},\ol{y}}{R}\right).
$
\end{minipage}
\end{center}

\medskip
We just mention here the two other types of foreign key constraints defined in SQL:1999, which could be captured by \FOLe formulae as well.
In the \emph{partial match} foreign key, for each row $r$ of the referencing table $R$, there must be some row $s$ in the referenced table $S$ such that the value of each referencing column $f_i$ is null or equal to the value of the corresponding referenced column $u_i$ in $s$. In the \emph{full match} foreign key, for each row $r$ of the referencing table $R$, either the value of each referencing column $f_i$ must be a null value or there must be some row $s$ in the referenced table $S$ such that the value of each referencing column $f_i$ is equal to the value of the corresponding referenced column $u_i$ in row $s$.

We observe that, if the database does not contain null values, the SQL:1999 definitions of unique, not null, and foreign key constraints (with simple match) involving null values reduce to the well known classical \FOL definitions of these constraints without null values.

\medskip
SQL:1999 introduced the possibility of specifying alternative semantics for the constraints (not based on the actual evaluation of SQL queries). In particular, foreign keys may be given optionally a semantic where null values would be interpreted as unknown values as opposed to nonexistent (see~\citep{TuGe:VLDB:01a}). The formalisation of this alternative semantics of null values only within constraints (but not queries) will require some more work. Note that commercial relational database systems do not support this extension.

\section{Discussion and Further Work}\label{sec:conclusions} 

We have shown the first-order logical characterisation of the \fosql fragment of SQL with null values, by first focussing on \nRA, a simple extension with null values of standard relational algebra -- which captures \fosql -- and then proposing two different domain relational calculi equivalent to \nRA, in which the null value does not appear as an element of the semantic interpretation domain of the logics. We have extended Codd's theorem by stating the equivalence of relational algebra with both domain relational calculi in presence of SQL null values.

In both cases it is clear that null values do not introduce incompleteness in a database in the sense of~\citet{Imielinski:1984}, since null values are just markers for inapplicable values. Indeed, a database in both calculi is represented as a unique first-order finite relational interpretation, and not as a set of interpretations, which would characterise the incompleteness of the database. In one calculus, a SQL table can be seen as a set of \emph{partial tuples}, while in the other (equivalent) calculus, a SQL table can be seen as a set of \emph{horizontally decomposed tables} each one holding regular total tuples.
This latest characterisation is in standard classical first order logic -- without a special ``null value'' domain element in the interpretation. This allows us to connect in a well founded way SQL databases represented in the decomposed relational calculus with knowledge representation formalisms based on classical logics: we could query SQL databases with null values mediated by ontologies in description logics or OWL \citep{DBLP:journals/jair/FranconiKN13,DBLP:conf/semweb/ArtaleFPS17} or state constraints and dependencies over SQL databases with null values using description logics or OWL \citep{10.1007/978-3-642-35176-1_28}. We plan to further explore the possibilities to use classical knowledge representation languages as conceptual modelling languages and ontology mediating languages for SQL databases with null values.

\medskip
Given our formal results, we can summarise the meaning of SQL null values by saying that they are just \emph{markers for inapplicable values within a tuple}. But how does the meaning of SQL null values relate to the \emph{inapplicable semantics} as it appears in the literature mentioned in Section~\ref{sec:intro}? 
According to this semantics, there should never be in the same table a tuple equal to another tuple but with some values substituted with null values; we have observed in Section~\ref{sec:intro} that this may happen in SQL. The \emph{inapplicable semantics} for null values imposes a constraint at table level, while SQL does not. 
Indeed, the execution of the projection and the union operators in a SQL query may lead to an inconsistent result if we adopt the inapplicable semantics for null values. A possible semantically sound solution to this problem~\citep{gottlob_closed_1988,roth_null_1989}, in the case of projection, is to require that the projected attributes always contain the primary key of each relation; in the case of union, is to require a stronger notion of union compatibility, verifying that the primary key is the same in each involved relation and that it is preserved in the union relation. \citet{lerat:lipski:86}, on the other hand, require to automatically eliminate all tuples with inapplicable nulls which are responsible for inconsistencies from the result of a projection or union, by \emph{deduplicating} the outcome of the operation modulo subsumption under inapplicable semantics.  Please note that similar requirements could be stated in the case of \emph{bag semantics}. It would be interesting to further analyse the relation between the semantics of SQL null values and the inapplicable semantics of null values.

\medskip

We find also interesting to relate the meaning of SQL null values with the \emph{no information semantics}. We give some preliminary hints here.

\noindent
It is possible to give a \emph{no information} meaning to null values in a \FOLe formula $\varphi$ with null values, by means of the following translation: each atom in the \FOLe formula with one or more null values is replaced by a disjunction whose disjuncts are the atom itself and the atoms with all combinations of existentially quantified variables in place of the null values. For example, in order to give a no information meaning to the null values in the formula $\varphi=P(t,\nullvalue,\nullvalue)$, we translate it to the \FOLe formula $\IFOLnin{\varphi}=P(t,\nullvalue,\nullvalue)\lor\exists x. P(t,x,\nullvalue)\lor\exists x. P(t,\nullvalue,x)\lor\exists x,y. P(t,x,y)$, with $x,y$ fresh new variables.  

\noindent
The \emph{no information} semantics of a database instance $\I$ with no information nulls is defined by a set \IFOLnin{\I} of \IFOLe{\I} instances, namely \IFOLnin{\I} denotes an incomplete database in the spirit of \citet{Imielinski:1984}, and we have to employ certain answers. \IFOLnin{\I} includes all the instances \IFOLe{\I} representing a possible realisation of a no information null value in $\I$ as a null value \nullvalue or as an arbitrary actual value of the domain. 

\noindent
We can show the following: for each \FOLe formula $\varphi$, instance $\I$, and assignment $\alpha$, if $\IFOLe{\I},\alpha \models \varphi$ then for all $J\in\IFOLnin{\I}$ it holds $\J,\alpha \models \IFOLnin{\varphi}$. The converse is not true, as example query (2) in Section~\ref{sec:intro} shows under no information semantics. However, if $\varphi$ does not have any null value, then the converse holds.

\bibliographystyle{elsarticle-harv-nourl}
\bibliography{null-values-bib}

\end{document}